\documentclass[aps,prl,twocolumn,superscriptaddress,nofootinbib]{revtex4-2}

\usepackage{amsmath,amssymb,bm}
\usepackage{graphicx}
\usepackage{hyperref}
\usepackage[nameinlink]{cleveref}
\usepackage{siunitx}
\usepackage{braket}
\usepackage{physics}
\usepackage{bbm}
\usepackage{xcolor}
\usepackage{amsthm}
\usepackage{appendix}
\usepackage{amssymb}

\usepackage[T1]{fontenc}
\usepackage{comicneue}

\usepackage{tikz}
\usetikzlibrary{arrows.meta,positioning,calc}

\hypersetup{colorlinks=true,linkcolor=blue,citecolor=blue,urlcolor=blue}
\crefname{equation}{Eq.}{Eqs.}
\Crefname{equation}{Equation}{Equations}
\crefname{figure}{Fig.}{Figs.}
\Crefname{figure}{Figure}{Figures}
\crefname{section}{Sec.}{Secs.}
\Crefname{section}{Section}{Sections}
\crefname{appendix}{Appendix}{Appendices}
\Crefname{appendix}{Appendix}{Appendices}

\newcommand{\hidmin}{n_{d}^{\min}}

\renewcommand{\b}{\beta}

\newtheorem{theorem}{Theorem}
 
\newtheorem{lemma}{Lemma} 
\newtheorem{definition}{Definition}

\begin{document}

\title{Information-Theoretic Scaling Laws of Neural Quantum States}
\author{Yiming Lu}
\affiliation{Department of Physics, Tsinghua University, Beijing 100084, China}
\author{Sriram Bharadwaj}
\thanks{Equal contribution}
\affiliation{
Mani L. Bhaumik Institute for Theoretical Physics, Department of Physics and Astronomy,\\ University of California, Los Angeles, CA 90095, USA
}
\author{Dikshant Rathore}
\thanks{Equal contribution}
\affiliation{
Mani L. Bhaumik Institute for Theoretical Physics, Department of Physics and Astronomy,\\ University of California, Los Angeles, CA 90095, USA
}
\author{Di Luo}
\email{diluo1000@gmail.com}
\affiliation{Department of Physics, Tsinghua University, Beijing 100084, China}
\affiliation{Institute of Advanced Study, Tsinghua University, Beijing 100084, China}
\date{\today}

\begin{abstract}
We establish an information-theoretic scaling law for generic autoregressive neural quantum states, determined by the middle-cut mutual information of the wavefunction amplitude.
By formalizing the virtual bond as an effective information channel across a sequence bipartition, we rigorously prove that exact autoregressive representation of a quantum state requires the virtual-bond dimension to scale with the amplitude mutual information. For stabilizer-state families, we show that this law yields an explicit, analytical rank formula. Applying this framework across quantum-state tomography, ground-state and finite-temperature learning, our numerical experiments expose precise exponent matching, architecture-dependent scaling differences between recurrent and Transformer neural quantum state, and the critical role of autoregressive basis ordering. These results establish a rigorous physical link between the intrinsic structure of a quantum many-body state and the corresponding neural-network capacity required for its faithful representation.
\end{abstract}

\maketitle
\textit{Introduction.}
Quantum many-body simulation lies at the heart of modern physics, underpinning our understanding of quantum materials, strongly correlated systems, and quantum information processing. However, the exponential growth of Hilbert space poses a fundamental challenge to classical computation, making efficient representations of quantum states a central problem. In recent years, neural-network quantum states (NQS) have emerged as a powerful and flexible framework for addressing this challenge \cite{carleo2017science,carleo2019rmp,deng2017prx,gao2017natcomm}. These approaches enable a broad range of applications, including quantum-state tomography \cite{torlai2018natphys,torlai2018prl}, variational ground-state search \cite{sharir2020prl,hibatallah2020prr,vicentini2022netket}, finite temperature \cite{nys2025fermionic}, real-time dynamics \cite{schmitt2020prl}, and electronic-structure calculations \cite{pfau2020prr,hermann2020natchem,Pescia_2024,Pescia_2022,kim2023neuralnetworkquantumstatesultracold, zhang2025neural,luo2019backflow}. Notably, in regimes characterized by frustration, topological order, and strong entanglement, NQS have demonstrated remarkable expressive power and scalability, extending the reach of classical methods for quantum many-body problems \cite{glasser2018prx,passetti2023entanglementtransitiondeepneural,PhysRevLett.134.076502,Choo_2019,Li_2022,Wu_2023,Hibat_Allah_2023,luo2023gauge,luo2021gauge,martyn2023variational}.

Despite their empirical success, a quantitative understanding of representability in neural-network quantum states remains limited. In practice, both architecture and model size are typically selected through extensive empirical tuning, involving repeated sweeps over depth and hidden dimensions for each target system. This approach is computationally costly and conceptually unsatisfactory, as it does not reveal which intrinsic features of a quantum state control the required representation capacity. From a theoretical perspective, universal approximation theorems guarantee asymptotic expressivity for broad model classes \cite{cybenko1989,hornik1989}, but offer little guidance on finite-size scaling or resource requirements. More recently, a growing body of work has begun to explore expressivity bounds, entanglement constraints, and representability limits for neural quantum states \cite{yang2024when,paul2025bound,deng2017prx,Jreissaty_2026,denis2025comment,levine2019quantum,Sharir_2022,li2025representationalpowerselectedneural,Passetti_2023,Jia_2020,kumar2026unlearnable,Luo_2023}, drawing parallels to scaling phenomena observed in modern large-scale machine learning models. Yet, in contrast to the rapid empirical advances seen in large language models~\cite{kaplan2020scaling,chen2025l2mmutualinformationscaling,cagnetta2026deriving}, where scaling laws provide predictive guidance for model design, an analogous, operational principle that links quantum-state structure to the required neural-network capacity is still lacking.

In this work, we introduce an information-theoretic scaling principle for quantum-state representations based on wavefunction amplitude structure. Any faithful representation must reproduce the amplitude distribution in a chosen basis, making amplitude complexity a necessary component of representational cost. This perspective is also operationally essential, as sampling and Monte Carlo procedures depend directly on the induced amplitude distribution. We show that the middle-cut mutual information of the quantum state amplitude gives rise to a virtual-bond (VB) scaling law for autoregressive neural-network quantum states (ARNN-NQS), which quantitatively links model capacity to the amplitude mutual information. For stabilizer-state families, we derive an explicit analytic scaling expressions in terms of closed-form rank formulas. We validate these predictions across three canonical settings: supervised tomography for tunable stabilizer families, ground-state learning for the toric code, and finite-temperature learning for thermofield-double states. Numerical experiments reveal precise exponent matching under controlled variations of target complexity, architecture-dependent scaling differences between recurrent and Transformer-based NQS, and a strong dependence of the required model size on basis ordering. Together, these results establish an information-theoretic principle that determines the necessary neural-network capacity required for faithful autoregressive representation of quantum many-body states from their intrinsic structure.

\begin{figure*}[ht]
\centering
\includegraphics[width=\textwidth]{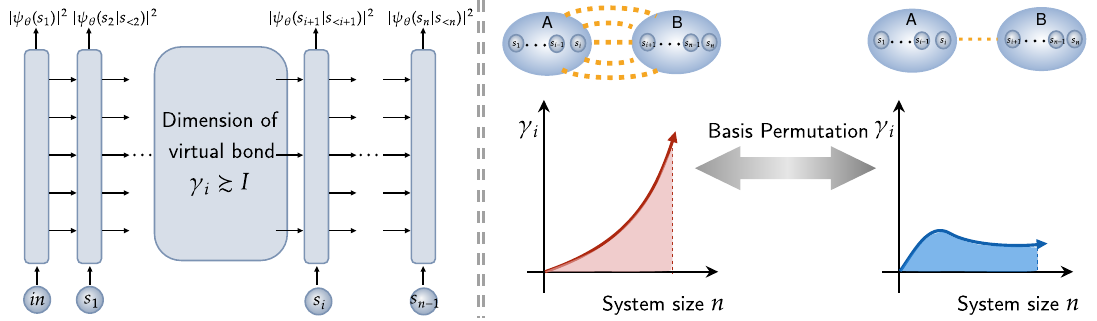}

\caption{Virtual-bond scaling law for ARNN-NQS.
Left: in an ARNN-NQS, information from the prefix $\bm s_{\leq i}$ is passed across the cut through a virtual bond, its effective dimension $\gamma_i$ quantifies the cross-cut information retained by the model.
Right: the scaling required of $\gamma_i$ is set by the amplitude complexity of the target family in the chosen basis.
A basis permutation can possibly change the required virtual-bond dimension dramatically. For example, converting a target with large $\gamma_i$ scaling into one with smaller or even $O(1)$ scaling.
The middle-cut mutual information of the quantum state amplitude directly determines the virtual-bond dimension required for faithful representation.}

\label{fig:nd_models}
\end{figure*}

\textit{Autoregressive neural quantum states.}
An ARNN-NQS represents a many-body wavefunction sequentially,
\begin{equation}
\psi_\theta(\bm s)=\prod_{i=1}^{n}\psi_\theta(s_i|\bm s_{<i}),
\qquad
P_\theta(\bm s)=|\psi_\theta(\bm s)|^2.
\end{equation}
This representation supports efficient sequential sampling \cite{sharir2020deep} and underlies many state-of-the-art applications of neural quantum states, including quantum tomography, ground state, and dynamics simulations \cite{torlai2018natphys,torlai2018prl,sharir2020prl,hibatallah2020prr,schmitt2020prl,vicentini2022netket,carrasquilla2021probabilistic,luo2023gauge,chen2023antn,luo2022autoregressive}.
Yet the required model capacity is still usually determined empirically, by sweeping architecture and width for each target family.
Such a procedure is costly and, more importantly, does not identify which structural features of the target state force the representation cost to grow.
Our goal is to connect target-state complexity directly to the model capacity required by exact ARNN-NQS representations.

Any exact ARNN-NQS representation must in particular reproduce the target amplitude distribution.
To quantify the cross-cut capacity required for this task, we introduce the \textit{virtual bond}, which captures the information that must be transmitted across a sequence cut in order to predict the future from the past.

\begin{definition}[ARNN-NQS Virtual Bond]
At cut position $i$, a variable $z_i$ is called a valid virtual bond for an ARNN-NQS if
\begin{equation}
|\psi_\theta(\bm s_{> i}\!\mid\!\bm s_{\leq i})|^2
=
|\psi_\theta(\bm s_{> i}\!\mid\!z_i)|^2.
\end{equation}
Here $z_i$ serves as a sufficient statistic of $\bm s_{\leq i}$ and is assumed to take values in $\mathbb R^{\gamma_i}$.
We denote by $\gamma_i$ the dimension of the virtual bond $z_i$.
\end{definition}

Common ARNN-NQS architectures include recurrent models, masked-convolution autoregressive networks, and attention-based Transformers \cite{larochelle2011nade,oord2016pixelcnn,cho2014properties,vaswani2017attention}.
Here we focus on recurrent networks and autoregressive Transformers as two representative memory organizations: one based on a compressed recurrent bottleneck, and the other on an explicitly growing cache.
We now identify their virtual bonds and the associated scaling of $\gamma_i$.

For a RNN NQS, the recurrent update takes the form
\begin{equation}
h_i=F_\theta(h_{i-1},s_i),\qquad
|\psi_\theta(s_{i+1}\!\mid\!\bm s_{\le i})|^2=G_\theta(h_i).
\end{equation}
The virtual bond is therefore the hidden state, $z_i=h_i$, and its effective dimension scales as $\gamma_i= n_d$. In this work, we choose 1D RNN NQS as a representative, which is referred to as RNN throughout the remainder of the paper unless otherwise specified.

For an autoregressive Transformer NQS, each token is mapped to query--key--value vectors, and the next prediction attends to the cached prefix memory:
\begin{equation}
\begin{aligned}
(q_i,k_i,v_i)&=\text{QKV-Map}(s_i),\\
|\psi_\theta(s_{i+1}\!\mid\!\bm s_{\le i})|^2&=\text{FF}(\mathrm{ATTN}\!\left(q_i,k_{\le i},v_{\le i}\right)).
\end{aligned}
\end{equation}
The virtual bond is therefore the present query and the key--value cache, i.e., $z_i=(q_i, k_{\le i},v_{\le i})$.
At fixed hidden width $n_d$, each cached token contributes a constant-size block, so the effective dimension grows linearly with $i$.
Explicit architecture-to-VB mappings are given in \cref{app:arnn-vb-definition}.

\textit{Quantum state amplitude complexity.}
Having introduced the virtual bond $z_i$, we now define the target-side complexity that this bottleneck must support.
Fix a local product basis $\mathcal S$ and an ordering of sites.
For a target $n$-qubit state $\ket{\psi_n}$, define the associated amplitude distribution
\begin{equation}
P_n(\bm s)=|\!\braket{\bm s}{\psi_n}\!|^2.
\end{equation}
At the middle cut $m=\lfloor n/2\rfloor$, partition the variables into $A=\bm s_{\leq m}$ and $B=\bm s_{> m}$, and define
\begin{equation}
\mathcal I(n)=I_{P_n}(A\!:\!B),
\end{equation}
where $I_P(A\!:\!B)=H(A)+H(B)-H(AB)$ and $H(\cdot)$ is the Shannon entropy.
We refer to this middle-cut classical mutual information (CMI) as the quantum-state amplitude complexity.

This quantity has an immediate operational interpretation for tensor-network representations.
For a matrix-product state (MPS), the middle-cut CMI is upper-bounded by the bipartite entanglement across the same cut, which is itself bounded by the logarithm of the bond dimension:
\begin{equation}
\mathcal I(n)\le 2S_{\mathrm{ent}}(n)\le 2\log_2 D_n.
\end{equation}
Thus any growth of the amplitude CMI with system size forces a corresponding growth of the MPS bond dimension.
Inspired by this observation, we derive the following scaling law for ARNN-NQS.

\begin{theorem}[VB Scaling Law]
Consider a target family in a fixed basis $\mathcal S$, and let
$m=\lfloor n/2\rfloor$ with $A=\bm s_{<m}$ and $B=\bm s_{\ge m}$.
If an autoregressive neural quantum state represents the corresponding amplitude distributions exactly across system sizes, then
\begin{equation}
\gamma_m\gtrsim \mathcal I(n) 
\end{equation}
\end{theorem}

The theorem follows from a rank bound on the middle-cut joint-probability matrix $Q_{AB}$ induced by the quantum state.
A finite virtual bond imposes a constraint on the row structure of this matrix, and, together with the inequality $I_{P_n}(A\!:\!B)\le \log_2 \operatorname{rank}(Q_{AB})$, yields the stated scaling law.
Detailed proofs under finite-precision and Lipschitz assumptions are given in \cref{app:vb-law-proof}.

\Cref{fig:nd_models} summarizes the VB scaling law.
For a chosen target family and autoregressive ordering, one first computes the middle-cut CMI complexity $\mathcal I(n)$.
The theorem then converts this target-side quantity into a representability requirement on the autoregressive channel $\gamma_m \gtrsim \mathcal I(n)$.
In particular, if $\mathcal I(n)\sim n^\beta$, then the required virtual-bond dimension must scale with the same exponent.
The middle-cut CMI complexity therefore becomes a direct scaling criterion for ARNN-NQS design.

We further prove that stabilizer states provide an analytically tractable class for which the required virtual-bond scaling can be made explicit.
For an ordered stabilizer family, the middle-cut CMI complexity in the computational basis admits a closed-form rank expression.
Combined with the VB scaling law, this immediately yields a concrete criterion for the minimum virtual-bond growth required by any exact autoregressive representation.
We state the result here and defer the derivation to \cref{app:stab-cmi}.

\begin{theorem}[ARNN-NQS VB Scaling Law for Stabilizer State Representation]
\label{cor:stabilizer-vb-scaling}
Let $\ket{\psi_S}$ be an ordered $n$-qubit stabilizer state.
At the middle cut $m=\lfloor n/2\rfloor$, define $A=\bm s_{<m}$ and $B=\bm s_{\ge m}$.
Let $M=(M_A,M_B)$ be the binary parity-check matrix of the $Z$-type stabilizers in the stabilizer table of $\ket{\psi_S}$, restricted to columns in $A$ and $B$.
Then any exact autoregressive representation of this quantum state must satisfy
\begin{equation}
\gamma_m \gtrsim
\operatorname{rank}(M_A)+\operatorname{rank}(M_B)-\operatorname{rank}(M),
\end{equation}
where all ranks are over $\mathbb F_2$.
The right-hand side is independent of the syndrome.
\end{theorem}

We next test the VB scaling law numerically in three canonical NQS settings: quantum tomography, ground-state learning, and finite-temperature learning.
Across controlled families with different power-law growth of $\mathcal I(n)$, the required RNN size follows the predicted exponents.
Direct comparisons between RNN and Transformer architectures further reveal architecture-dependent growth consistent with the theorem.
Guided by the same principle, we also show that basis permutations can reduce amplitude complexity and thereby lower learning difficulty without increasing model size.

\textit{Quantum tomography.}
We first test the VB scaling law in supervised quantum tomography, a standard application of neural-network quantum states \cite{torlai2018neural}.
Stabilizer states provide a controlled benchmark because their middle-cut amplitude complexity can be computed analytically and tuned across families.
Throughout this section, we use the stabilizer rank formula derived in \cref{app:stab-cmi}.

As a benchmark on an $L\times L$ lattice, we consider a checkerboard stabilizer family in which the number of parity checks scales as $L^{2\gamma}$.
For this family, the middle-cut amplitude complexity in the $Z$ basis scales as $
\mathcal I(L)\propto L^\gamma .$
An explicit construction and the corresponding scaling analysis are provided in \cref{app:checkerboard-construction}.

\begin{figure}[ht]
  \centering
  \includegraphics[width=\linewidth]{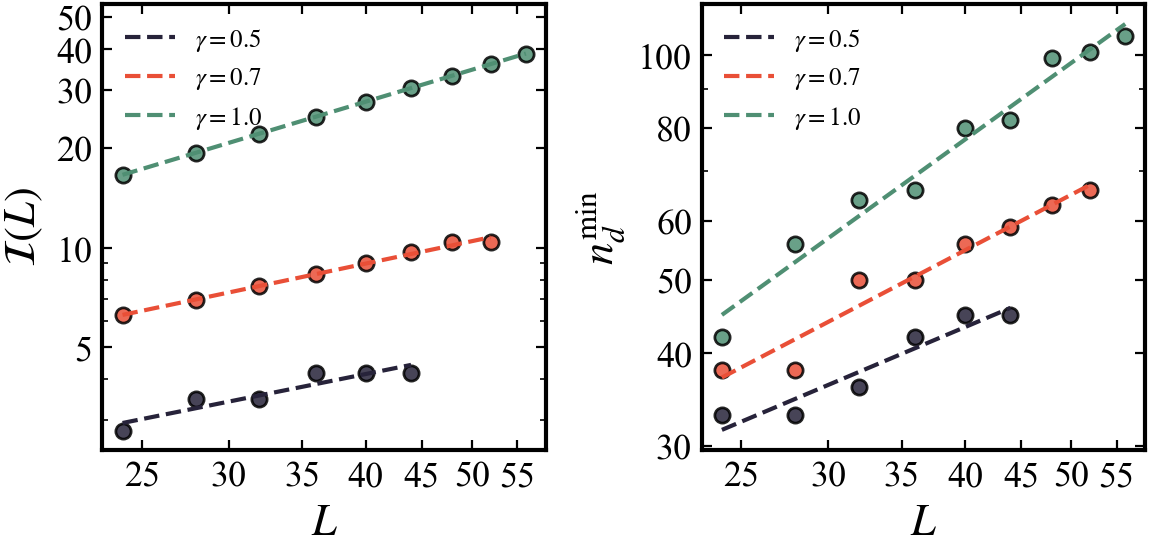}
  \caption{Application I: checkerboard-stabilizer tomography.
  Left, log-log plot of the half-cut CMI $\mathcal I(L)$ versus $L$ for tunable families with number of parity checks scales as $L^{2\gamma}$; fitted slopes are $0.670$, $0.702$, and $1.000$ for $\gamma=0.5,0.7,1.0$.
  Right, log-log plot of the minimal RNN hidden width $\hidmin(L)$ required to reach $95\%$ fidelity; fitted slopes are $0.620$, $0.767$, and $1.057$.
  The close agreement between the two sets of exponents shows that the required RNN VB dimension tracks the scaling of the target amplitude complexity.}
  \label{fig:graph_tuning}
\end{figure}

For each family, we perform supervised quantum tomography with an RNN and extract the minimal hidden width $\hidmin$ required to reach $95\%$ fidelity.
Since $\gamma_i=n_d$ for an RNN, $\hidmin(L)$ serves as a direct proxy for the minimum virtual-bond dimension.
As shown in \Cref{fig:graph_tuning}, the left panel exhibits clearly separated power-law growth of $\mathcal I(L)$ for the three tunable families, with larger $\gamma$ producing systematically steeper complexity scaling.
The right panel displays the same hierarchy in the required RNN width: the $\gamma=1.0$ family grows fastest, the $\gamma=0.7$ family is intermediate, and the $\gamma=0.5$ family grows slowest.
The fitted exponents of $\hidmin(L)$ closely track those of $\mathcal I(L)$ across all tested $\gamma$, in quantitative agreement with the VB scaling law.

\textit{Ground-state learning.}
We next test the VB scaling law in a ground-state learning task.
This setting isolates the role of architecture: the target-side scaling is fixed, while the virtual-bond organization differs between RNNs and autoregressive Transformers.
It provides insight for comparing how different autoregressive architectures realize the required virtual-bond dimension for the same target state.

We consider the $\mathbb Z_2$  $L\times L$ toric code with Hamiltonian
\begin{equation}
H_{\mathrm{TC}}=-\sum_s A_s-\sum_p B_p,
\end{equation}
where $A_s=\prod_{e\in +_s}X_e$ and $B_p=\prod_{e\in\partial p}Z_e$.
The ground-state is fixed to a single state by adding the two non-contractible $X$-loop operators. In the $Z$ basis, this state takes the form
\begin{equation}
\ket{\psi_{\mathrm{TC}}}
=\frac{1}{\sqrt{|\Omega|}}\sum_{\bm s\in\Omega}\ket{\bm s},
\qquad
P^{(Z)}_{\mathrm{TC}}(\bm s)=\frac{1}{|\Omega|}\,\mathbf 1_{\bm s\in\Omega},
\label{eq:toric-uniform}
\end{equation}
where $\Omega$ is the set of configurations satisfying all stabilizer parity constraints.
These constraints couple the two sides of a cut and therefore produce linear growth of the middle-cut amplitude complexity.

The stabilizer rank formula can be specialized to the toric code, as shown in \cref{app:toric-cmi}.
For a half-system cut, one obtains $\mathcal I(L)=\Theta(L)$
with the explicit finite-size form $\mathcal{I}(L)=2L-1$ for even $L$ and $2L$ for odd $L$.
By lattice $C_4$ symmetry, the same scaling holds for both vertical and horizontal cuts.

\begin{figure}[ht]
  \centering
  \includegraphics[width=\linewidth]{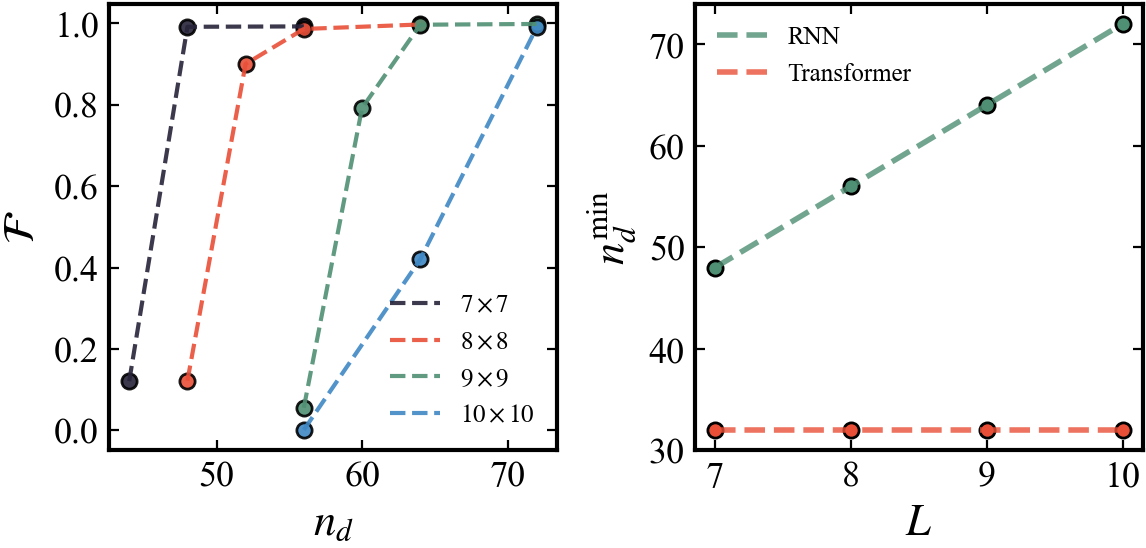}
  \caption{Application II: toric-code benchmark.
  Left, fidelity versus hidden width $n_d$ for RNN training on system sizes from $7\times7$ to $10\times10$.
  Right, minimal width $n_d^{\min}$ versus linear size $L$ at fixed target fidelity ($95\%$) for RNN and autoregressive Transformer models.
  The required RNN width grows with system size, whereas the Transformer width remains nearly constant over the tested range, consistent with the different scaling of virtual-bond dimension in recurrent and attention-based architectures.}
  \label{fig:toric_vmc}
\end{figure}

The VB scaling law therefore predicts that exact recurrent representations of this family require growing virtual-bond dimension with $L$.
The numerics in \Cref{fig:toric_vmc} support this picture.
In the left panel, the RNN fidelity curves shift systematically toward larger hidden width as the system size increases, indicating that a larger recurrent bottleneck is needed to represent the toric-code accurately.
The right panel summarizes this trend through the extracted minimum width $n_d^{\min}$ at fixed target fidelity.
For RNNs, $n_d^{\min}$ grows clearly with system size, reflecting the need for an expanding recurrent virtual bond.
By contrast, autoregressive Transformers achieve comparable fidelity with nearly constant width over the tested range, because their effective virtual-bond dimension also grows with sequence length through the key--value cache.

\textit{Finite-temperature representation.}
We finally test the VB scaling law in finite-temperature learning, an important application of neural-network quantum states \cite{irikura2020neural,nys2024real,nys2025fermionic,kumar2026autoregressive}.
Here the target family is fixed, while the autoregressive ordering can be varied, allowing a direct test of how basis ordering reshapes the mutual information complexity and hence the required virtual-bond scaling.
As a temperature-tunable example, we study thermofield-double (TFD) states in a fixed computational basis.
For a Hamiltonian $H$, the TFD state is
\begin{align}
\ket{\mathrm{TFD}_\beta}
&= \frac{1}{\sqrt{Z(\beta)}}\sum_{\mu} e^{-\beta E_\mu/2}\,\ket{\mu}_A\otimes\ket{\mu}_B,
\end{align}
where $Z(\beta)=\sum_\mu e^{-\beta E_\mu}$.
In the doubled computational basis $(a,b)\in \mathcal{H}_A\otimes\mathcal{H}_B$, the dephased probability distribution is
\begin{equation}
\mathbb P_\beta(a,b)=\frac{\left|\bra{a}e^{-\beta H/2}\ket{b}\right|^2}{Z(\beta)}.
\end{equation}

We instantiate $H$ as the spinless $p$-wave BCS chain,
\begin{equation}
H
=\sum_{j=1}^{n} \left[-J\left(c_j^\dagger c_{j+1} + c_j^\dagger c_{j+1}^\dagger\right) +2h\left(c_j^\dagger c_j-\tfrac12\right)\right],
\end{equation}
At $\beta=0$, one has $\mathbb P_0(a,b)=2^{-n}\delta_{ab}$, so each site pair $(a_i,b_i)$ is maximally correlated.
For small finite $\beta$, the expansion $e^{-\beta H/2}=\mathbb I-(\beta/2)H+O(\beta^2)$ implies that off-diagonal weights remain perturbative, with $\mathbb P_\beta(a\neq b)=O(\beta^2)$.
Thus, at small finite $\beta$, the middle-cut CMI is still controlled mainly by the copy-to-copy pair correlations inherited from $\beta=0$.
It is therefore large when the ordering places many such pairs across the cut, and small when it keeps them local in sequence space.
Additional analytical perturbative calculation and numerical results in finite and high temperature regimes are given in \cref{app:TFIM}.

The above analysis motivates two autoregressive orderings of the doubled bitstring,
\begin{equation}
\begin{aligned}
\text{separate:}\quad & a_1a_2\cdots a_n\,b_1b_2\cdots b_n,\\
\text{alternate:}\quad & a_1b_1a_2b_2\cdots a_nb_n.
\end{aligned}
\end{equation}
The alternate ordering places the dominant pair correlations locally in sequence space, whereas the separate ordering spreads them over long distances.
The VB scaling law therefore predicts sharply different virtual-bond dimension requirements: near-linear growth for the separate ordering and near-constant behavior for the alternate ordering over the tested regime.

\begin{figure}[ht]
  \centering
  \includegraphics[width=\linewidth]{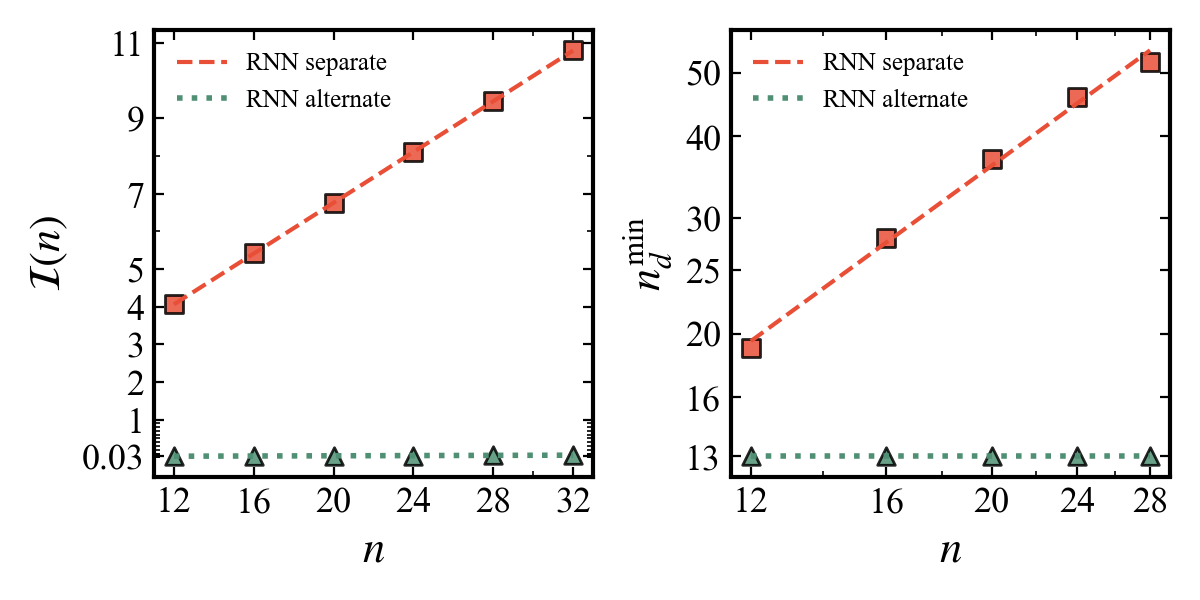}
  \caption{Application III: finite-temperature representation at $J=1, h=0.6,\beta=0.1$,
  Left, plot of the half-cut CMI versus $n$ for two orderings of the doubled basis.
  The numerical method for the CMI calculation is shown in \Cref{app:TFD-fermion-amplitude}.
  The separate ordering is nearly linear (fitted slope $0.997$, $\mathcal{I}\propto n^{0.997}$), whereas the alternate ordering remains small and nearly constant.
  Right, minimal RNN hidden width $n_d^{\min}$ at fixed target fidelity for the same two orderings.
  For the separate ordering, $n_d^{\min}$ grows with $n$ (fitted slope $1.205$, $n_d^{\min}\propto n^{1.205}$); for the alternate ordering, the target fidelity is reached with nearly constant hidden width over the tested range.}
  \label{fig:finite-temp}
\end{figure}

We test this prediction by RNN training at fixed target fidelity and extract the minimal hidden width $n_d^{\min}$ for each system size.
As shown in \Cref{fig:finite-temp}, the two orderings exhibit sharply different behavior in both amplitude complexity and required hidden width.
At $\beta=0.1$, the separate ordering yields near-linear growth of $\mathcal I(n)$ and correspondingly increasing $n_d^{\min}$.
By contrast, the alternate ordering keeps $\mathcal I(n)$ small and nearly constant, and reaches the target fidelity with nearly constant hidden width over the tested range. This result highlights the crucial role of efficient autoregressive ordering in finite-temperature simulations.

\textit{Conclusion.}
We have established an information-theoretic scaling law that governs the representational power of neural quantum states. By identifying the middle-cut mutual information of the quantum state amplitude as a fundamental complexity measure, we derive a rigorous virtual-bond (VB) scaling law for exact autoregressive representations of quantum states. This framework elevates neural-network design from empirical tuning to a physically grounded principle, directly linking the intrinsic correlation structure of a quantum many-body state to the minimal network capacity required for its faithful representation. Analytical results for stabilizer states, together with numerical validations across quantum-state tomography, topological ground states, and finite-temperature learning, confirm the predicted scaling behavior. These results further reveal differences between RNN and Transformer neural quantum state, and demonstrate how autoregressive ordering reshapes the representation cost.

These findings lay the foundation for a predictive complexity theory of neural quantum states. It paves the way for systematic analysis across a broad class of architectures, including 2D RNNs \cite{hibat2020recurrent}, PixelCNN \cite{sharir2020deep}, and Vision Transformer–based NQS \cite{viteritti2023transformer}. While the present scaling law governs autoregressive constructions, extending this framework to more general neural quantum states, such as neural-network backflow~\cite{luo2019backflow} and real-space wavefunction representations~\cite{pfau2020prr}, remains an important direction. More fundamentally, developing a unified theory of neural quantum state representation that treats both amplitude and phase complexity is an open challenge. Beyond representability in equilibrium systems, integrating these information-theoretic constraints with quantum learning dynamics could provide a principled route toward scaling neural-network approaches to increasingly complex quantum many-body systems.

\textit{Acknowledgements.} The authors acknowledge insightful discussions with Zhuo Chen. SB and DR gratefully acknowledge the Mani L. Bhaumik Institute for Theoretical Physics for supporting this work.
\bibliographystyle{apsrev4-1}
\bibliography{refs.bib}

\newpage

\appendix 
\crefalias{section}{appendix}

 \clearpage

\onecolumngrid
\begin{center}
\noindent\textbf{Supplementary Material}
\bigskip

\noindent\textbf{\large{}}
\end{center}
\onecolumngrid
\setcounter{secnumdepth}{2}
\renewcommand{\thefigure}{S\arabic{figure}}
\section{Virtual-Bond Definition and Architecture Mapping}
\label{app:arnn-vb-definition}
This appendix makes the virtual-bond variable in the VB Scaling Law explicit and ties it to the two autoregressive architectures used in the numerics.
An ARNN-NQS parameterizes amplitudes as
\begin{equation}
|\psi_\theta(\bm s)|^2=\prod_{i=1}^n |\psi_\theta(s_i|\bm s_{<i})|^2.
\end{equation}
At a cut position $i$, a variable $z_i$ is a valid \emph{virtual bond} if it is sufficient for suffix prediction:
\begin{equation}
|\psi_\theta(\bm s_{> i}|\bm s_{\leq i}|^2=|\psi_\theta(\bm s_{>i}|z_i)|^2.
\end{equation}
Here $z_i$ is assumed to take values in $\mathbb{R}^{\gamma_i}$ and $\gamma_i$ is denoted as the virtual-bond dimension.
This variable is called a virtual bond because it plays the same structural role as a bond index in tensor-network factorizations.
Once the cut is fixed, all dependence between prefix and suffix is mediated through this internal channel.
Therefore $\gamma_i$ is the relevant virtual-bond dimension in the VB Scaling Law.
\Cref{structure-ARNN} summarizes this channel in the RNN and autoregressive Transformer used in this work.

\begin{figure}[ht]
    \centering
\begin{tikzpicture}[
    scale=0.75,
    transform shape,
    x=0.92cm,y=0.92cm,
    font=\sffamily,
    >=Latex,
    line width=0.9pt,
    title/.style={font=\bfseries\fontsize{18}{20}\selectfont},
    smalllab/.style={font=\fontsize{10.5}{12}\selectfont},
    mathlab/.style={font=\fontsize{11}{13}\selectfont},
    boxpink/.style={draw=black, rounded corners=3pt, fill=pink!18, minimum width=0.8cm, minimum height=0.7cm, inner sep=0pt},
    boxpurple/.style={draw=black, rounded corners=4pt, fill=blue!18!violet!18, minimum width=1.55cm, minimum height=1.02cm, inner sep=0pt},
    boxyellow/.style={draw=black, rounded corners=4pt, fill=orange!35, minimum width=2.25cm, minimum height=1.02cm, inner sep=2pt},
    boxstack/.style={draw=black, rounded corners=6pt, fill=orange!35, minimum width=2.95cm, minimum height=2.5cm, inner sep=4pt},
    bubble/.style={draw=black, rounded corners=10pt, fill=gray!12, minimum width=4.3cm, minimum height=2.95cm, inner sep=6pt},
    innerpink/.style={draw=black, rounded corners=2pt, fill=magenta!20, minimum width=2.2cm, minimum height=0.56cm, inner sep=1pt},
    innerpurple/.style={draw=black, rounded corners=2pt, fill=violet!25, minimum width=2.2cm, minimum height=0.56cm, inner sep=1pt},
    inneryellow/.style={draw=black, rounded corners=2pt, fill=yellow!70!orange!40, minimum width=2.2cm, minimum height=0.56cm, inner sep=1pt},
]

\node[title] at (3.8,4.9) {\large RNN};
\node[title] at (14.2,4.9) {\large Autoregressive Transformer};

\node[boxpink]   (sleft) at (2.5,3.8) {$s_{i-1}$};
\node[boxpurple] (gru1)  at (2.5,2.3) {\Large Cell};
\node[boxpink]   (si)    at (5.5,3.8) {$s_i$};
\node[boxpurple] (gru2)  at (5.5,2.3) {\Large Cell};
\coordinate (gru3) at (7.5,2.3);
\coordinate (gru0) at (0.5,2.3);

\draw[->] (sleft) -- (gru1);
\draw[->] (si) -- (gru2);

\draw[->] (gru0) -- node[above,smalllab] {$h_{i-2}$} (gru1);
\draw[->] (gru1) -- node[above,smalllab] {$h_{i-1}$} (gru2);
\draw[->] (gru2) -- node[pos=0.4, above,smalllab] {$h_i$} (gru3);
\draw[->] (gru1) -- ++(0,-1.05) node[below,mathlab] {$p_\theta(s_{i}|s_{<i})$};
\draw[->] (gru2) -- ++(0,-1.05) node[below,mathlab] {$p_\theta(s_{i+1}|s_{<i+1})$};

\node[bubble, minimum width=5cm, minimum height=2.5cm] (rbub) at (4.1,-1.3) {};
\node[mathlab] at (4.1,-0.6) {$h_i \in \mathbb{R}^{n_d}$};
\node[boxpurple] at (3.9,-1.6) {\Large Cell};
\node[mathlab] at (5.62,-1.6) {$\sim O(n_d^2)$};
\node[mathlab] at (2.2,-1.6) {Params};

\node[boxpink]   (sctx)  at (11.7,3.8) {$s_{i-1}$};
\node[boxyellow] (enc1)  at (11.7,2.3) {\Large Encoder};

\node[boxpink]   (sit)   at (16.2,4.2) {$s_i$};
\node[boxstack, minimum height=2.0cm]  (stack) at (16.2,2.3) {};
\node[innerpink]   at (16.2,2.92) {QKV-Map};
\node[innerpurple] at (16.2,2.30) {ATTN};
\node[inneryellow] at (16.2,1.68) {FF};
\coordinate (enc3) at (19.3, 2.3);
\coordinate (enc0) at (8.8, 2.3);

\draw[->] (enc0) -- node[above, pos=0.49, mathlab] {$kv_{< i-1}$} (enc1.west);
\draw[->] (sctx) -- (enc1);
\draw[->] (enc1) -- node[above,mathlab, pos=0.45] {$kv_{<i}$} (stack.west);
\draw[->] (sit) -- (stack.north);
\draw[->] (enc1) -- ++(0,-1.05) node[below,mathlab] {$p_\theta(s_{i}|s_{<i})$};
\draw[->] (stack) -- ++(0,-1.53) node[below,mathlab] {$p_\theta(s_{i+1}|s_{<i+1})$};

\draw[->] (stack.east) -- node[above, pos=0.43, mathlab] {$kv_{\leq i}$} (enc3);

\node[bubble, minimum width=6cm, minimum height=2.6cm] (rbub1) at (14.1,-1.3) {};
\node[mathlab] at (14.2,-0.6) {$q_i, v_i, k_i \in \mathbb{R}^{n_d}$};
\node[boxyellow] at (14.0,-1.6) {\Large Encoder};
\node[mathlab] at (16.2,-1.6) {$\sim O(n_d^2)$};
\node[mathlab] at (11.9,-1.6) {Params};

\end{tikzpicture}
\caption{Virtual-bond realizations for the two autoregressive architectures used in numerics.
Left panel, RNN: the recurrent update transmits information through $h_{i-1}\!\to h_i$, so the recurrent state is the virtual bond.
Right panel, autoregressive Transformer: prefix information is carried by the cached key-value memory through $kv_{<i}\!\to kv_{\le i}$, so the cache is the virtual bond.
The figure gives the model-specific map from theorem-level dimension $\gamma_i$ to architecture size.}
\label{structure-ARNN}
\end{figure}
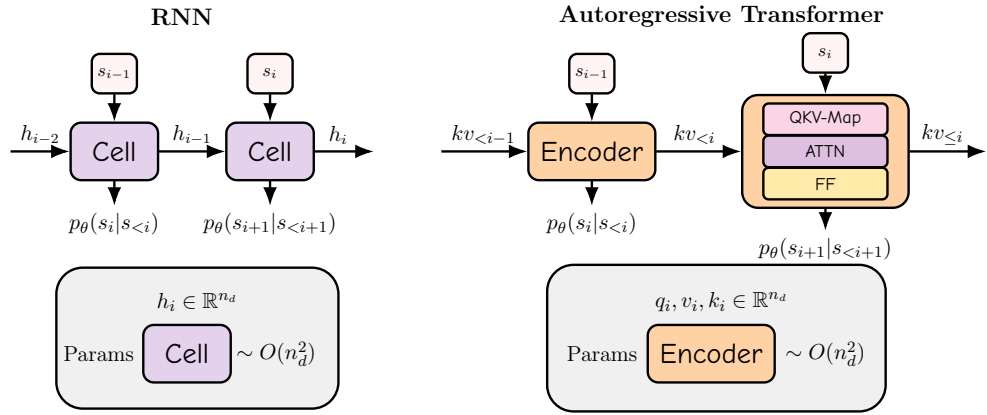

For the RNN family used in numerics, the recurrent update can be written as
\begin{equation}
h_i=F_\theta(h_{i-1},s_i),\qquad
p_\theta(s_{i+1}|\bm s_{\le i})=G_\theta(h_i).
\end{equation}
Hence prefix information reaches future predictions only through $h_i$.
The virtual bond is therefore
\begin{equation}
z_i=h_i\in\mathbb R^{n_d},
\end{equation}
so $\gamma_i=n_d$.

For the autoregressive Transformer family, let
\begin{equation}
\mathcal K_i=\{(k_1,v_1),\ldots,(k_i,v_i)\}
\end{equation}
denote the cached key-value history at cut $i$.
The next-token predictor depends on the prefix through this cache, so the virtual bond is $z_i=\{q_i, \mathcal K_i\}\in \mathbb{R}^{n_d\times (2i+1)}$.
At fixed depth and width, each cached token contributes a fixed number of coordinates proportional to $n_d$, while the number of cached tokens grows with $i$.
Therefore the effective virtual-bond dimension can grow with sequence length even at fixed width.
This is exactly the mechanism tested by the RNN and Transformer comparison in the main text.

\section{Proof of the VB Scaling Law}
\label{app:vb-law-proof}

We prove the theorem in the main text using the ordinary rank of the joint-distribution matrix.
Throughout, we fix the middle cut
\begin{equation}
m:=\lfloor n/2\rfloor,
\end{equation}
and write
\begin{equation}
A:=\{\bm s_{\le m}\},\qquad B:=\{\bm s_{>m}\},
\end{equation}
so that
\begin{equation}
\mathcal I(n)=I_{P_n}(A\!:\!B).
\end{equation}
Since the ARNN-NQS family is assumed to represent the target amplitude distribution exactly, we identify
\begin{equation}
P_\theta(\bm s)=P_n(\bm s)
\end{equation}
for each system size $n$.

By definition, a valid virtual bond $z_t$ must fully characterize the dependence of the future on the past.
At the middle cut $m$, this means
\begin{equation}
P_n(B\!\mid\!A)=P_\theta(B\!\mid\!A)=P_\theta(B\!\mid\!z_m(A)).
\label{eq:appendix-vb-sufficiency}
\end{equation}

We first recall a standard rank-to-mutual-information inequality.

\begin{lemma}[Mutual information is bounded by ordinary rank]
\label{lem:MI-rank-appendix}
Let $Q_{AB}$ be the joint-distribution matrix of two finite random variables $A,B$, namely
\begin{equation}
(Q_{AB})_{a,b}=P(a,b).
\end{equation}
Then
\begin{equation}
I(A\!:\!B)\le \log_2 \operatorname{rank}(Q_{AB}).
\label{eq:appendix-MI-rank}
\end{equation}
\end{lemma}

\begin{proof}
Let $p_A(a)=\sum_b P(a,b)$ and $p_B(b)=\sum_a P(a,b)$, and define
\begin{equation}
M_{ab}:=\frac{P(a,b)}{\sqrt{p_A(a)p_B(b)}}.
\end{equation}
Then $\operatorname{rank}(M)=\operatorname{rank}(Q_{AB})$.
Moreover, the largest singular value of $M$ is at most $1$, hence
\begin{equation}
\|M\|_F^2\le \operatorname{rank}(M)=\operatorname{rank}(Q_{AB}).
\end{equation}
On the other hand,
\begin{equation}
\chi^2\!\left(P_{AB}\,\middle\|\,p_Ap_B\right)
=
\sum_{a,b}\frac{P(a,b)^2}{p_A(a)p_B(b)}-1
=
\|M\|_F^2-1.
\end{equation}
Using
\begin{equation}
I(A\!:\!B)
=
D\!\left(P_{AB}\,\middle\|\,p_Ap_B\right)
\le
\log_2\!\Bigl(1+\chi^2(P_{AB}\|p_Ap_B)\Bigr),
\end{equation}
we obtain
\begin{equation}
I(A\!:\!B)\le \log_2 \|M\|_F^2\le \log_2 \operatorname{rank}(Q_{AB}).
\end{equation}
\end{proof}

\subsection{Finite precision assumption}

We first prove the main-text theorem under a finite-precision assumption on the virtual bond.

\begin{theorem}[VB Scaling Law, Finite Precision Assumption]
\label{thm:VB-finite-precision-rank}
Assume that any middle-cut virtual bond $z_m(A)\in\mathbb{R}^{\gamma_m}$, where $m=\lfloor n/2\rfloor$ and $A=\{s_1,\cdots,s_m\}$, is represented by $\gamma_m$ scalars with $b$ bits per scalar. Then
\begin{equation}
\mathcal I(n)\le b\gamma_m,
\label{eq:finite-precision-main}
\end{equation}
which implies
\begin{equation}
\gamma_m\ge \frac{\mathcal I(n)}{b}.
\end{equation}
Therefore, when the precision $b$ is fixed independently of $n$, $\gamma_m$ and $\mathcal I(n)$ have the same polynomial scaling exponent.
\end{theorem}

\begin{proof}
Let $Q_{AB}$ be the matrix representation of the joint distribution $P_n(A,B)$, namely
\begin{equation}
(Q_{AB})_{A,B}=P_n(A,B).
\end{equation}
Under the finite-precision assumption, the virtual bond $z_m$ can take values only in a finite set $\mathcal Z_m$ with cardinality
\begin{equation}
|\mathcal Z_m|=2^{b\gamma_m}.
\end{equation}

We now factorize $Q_{AB}$ through the intermediate virtual-bond index $z\in\mathcal Z_m$. Define
\begin{align}
M_{\mathrm{left}}(A,z) &:= P_n(A)\,\delta_{z,z_m(A)},\\
M_{\mathrm{right}}(z,B) &:= P_\theta(B\!\mid\!z).
\end{align}
Then
\begin{equation}
Q_{AB}=M_{\mathrm{left}}M_{\mathrm{right}},
\end{equation}
because for every $(A,B)$,
\begin{equation}
(Q_{AB})_{A,B}
=
\sum_{z\in\mathcal Z_m} M_{\mathrm{left}}(A,z)M_{\mathrm{right}}(z,B)
=
\sum_{z\in\mathcal Z_m} P_n(A)P_\theta(B\!\mid\!z)\delta_{z,z_m(A)}
=
P_n(A)P_\theta(B\!\mid\!z_m(A)),
\end{equation}
which is exactly the virtual-bond sufficiency relation.

Since the inner dimension of this factorization is $|\mathcal Z_m|$, we immediately obtain
\begin{equation}
\operatorname{rank}(Q_{AB})\le |\mathcal Z_m|.
\end{equation}
Applying Lemma~\ref{lem:MI-rank-appendix} gives
\begin{equation}
\mathcal I(n)=I_{P_n}(A\!:\!B)
\le \log_2\operatorname{rank}(Q_{AB})
\le \log_2|\mathcal Z_m|
= b\gamma_m.
\end{equation}
This proves Eq.~\eqref{eq:finite-precision-main}.
\end{proof}

\subsection{Lipschitz continuity assumption}

We next consider a continuous virtual bond.
In this case, the exact ordinary rank of the original joint-distribution matrix need not be controlled directly.
Instead, we construct a low-rank surrogate matrix by discretizing the virtual-bond space.

\begin{theorem}[VB Scaling Law, Lipschitz Assumption]
\label{thm:VB-lipschitz-rank-main}
Assume that any middle-cut virtual bond $z_m(A)\in\mathbb R^{\gamma_m}$, where $m=\lfloor\frac{n}{2}\rfloor$ and $A=\{s_1,\cdots,s_m\}$, satisfies:

\begin{enumerate}
    \item \textbf{Bounded bond space:} there exists a constant $\mathcal{R}$ such that
    \begin{equation}
    \|z_m(A)\|_2\le \mathcal{R}
    \qquad \text{for all }A.
    \end{equation}

    \item \textbf{Conditional-distribution Lipschitz continuity:}
    there exists a constant $L_{\mathrm{Lip}}>0$ such that
    \begin{equation}
    \|P_\theta(\cdot\!\mid\!z)-P_\theta(\cdot\!\mid\!z')\|_1
    \le
    L_{\mathrm{Lip}}\|z-z'\|_2
    \qquad \text{for all } z,z'\in \mathbb B_{\gamma_m}(\mathcal{R}),
    \label{eq:appendix-dist-lip}
    \end{equation}
    where $\mathbb B_{\gamma_m}(\mathcal{R}):=\{z\in\mathbb R^{\gamma_m}:\|z\|_2\le \mathcal{R}\}$ is the Euclidean ball of radius $\mathcal{R}$, and the norm is over the suffix variable $B$.
\end{enumerate}

Then for every $0<\varepsilon<1$, there exists a surrogate joint distribution $\widetilde P(A,B)$ and its corresponding matrix $\widetilde Q_{AB}$ such that
\begin{equation}
\sup_A \|P_n(\cdot\!\mid\!A)-\widetilde P(\cdot\!\mid\!A)\|_1\le \varepsilon,
\label{eq:appendix-rowwise-approx}
\end{equation}
and
\begin{equation}
\operatorname{rank}(\widetilde Q_{AB})
\le
\left(\frac{C\,\mathcal{R}\,L_{\mathrm{Lip}}}{\varepsilon}\right)^{\gamma_m},
\label{eq:appendix-rank-surrogate}
\end{equation}
for a geometric constant $C$.

Consequently,
\begin{equation}
\mathcal I(n)
\le
\gamma_m\log_2\!\left(\frac{C\,\mathcal{R}\,L_{\mathrm{Lip}}}{\varepsilon}\right)
+\varepsilon \log_2(|\mathcal B|-1)
+2h_2(\varepsilon/2),
\label{eq:appendix-lip-main-bound}
\end{equation}
where $|\mathcal B|$ is the support size of $B$ and
$h_2(x)=-x\log_2 x-(1-x)\log_2(1-x)$.
Hence
\begin{equation}
\gamma_m\ge
\frac{\mathcal I(n)-\Delta_\varepsilon(B)}
{\log_2\!\left(\frac{C\,\mathcal{R}\,L_{\mathrm{Lip}}}{\varepsilon}\right)},
\qquad
\Delta_\varepsilon(B):=\varepsilon \log_2(|\mathcal B|-1)+2h_2(\varepsilon/2).
\end{equation}
In particular, for fixed $\varepsilon$, $\mathcal{R}$, and $L_{\mathrm{Lip}}$ under a fixed approximation protocol, the denominator is $n$-independent and the correction term is controlled, so $\gamma_m$ and $\mathcal I(n)$ exhibit the same polynomial scaling exponent.
\end{theorem}

\begin{proof}
Set
\begin{equation}
\eta:=\varepsilon/L_{\mathrm{Lip}}.
\end{equation}
Since all virtual bonds lie in the Euclidean ball $\mathbb B_{\gamma_m}(\mathcal{R})$, let $\mathcal N_\eta$ be an $\eta$-net of this ball.
Standard covering-number bounds imply
\begin{equation}
|\mathcal N_\eta|
\le
\left(\frac{C\,\mathcal{R}}{\eta}\right)^{\gamma_m}
=
\left(\frac{C\,\mathcal{R}\,L_{\mathrm{Lip}}}{\varepsilon}\right)^{\gamma_m}.
\label{eq:appendix-covering}
\end{equation}

For each prefix $A$, let $\widehat z_m(A)\in\mathcal N_\eta$ be a nearest net point to $z_m(A)$, and define the surrogate distribution as
\begin{equation}
\widetilde P(B\!\mid\!A):=P_\theta(B\!\mid\!\widehat z_m(A)),
\qquad
\widetilde P(A,B):=P_n(A)\,\widetilde P(B\!\mid\!A).
\end{equation}
Since $\|z_m(A)-\widehat z_m(A)\|_2\le \eta$, the Lipschitz assumption gives
\begin{equation}
\|P_n(\cdot\!\mid\!A)-\widetilde P(\cdot\!\mid\!A)\|_1
=
\|P_\theta(\cdot\!\mid\!z_m(A))-P_\theta(\cdot\!\mid\!\widehat z_m(A))\|_1
\le
L_{\mathrm{Lip}}\eta
=
\varepsilon,
\end{equation}
which proves Eq.~\eqref{eq:appendix-rowwise-approx}.

Let $\widetilde Q_{AB}$ be the matrix representation of the surrogate joint distribution $\widetilde P(A,B)$, namely
\begin{equation}
(\widetilde Q_{AB})_{A,B}=\widetilde P(A,B).
\end{equation}
We now factorize $\widetilde Q_{AB}$ through the discretized virtual-bond index $\widehat z\in\mathcal N_\eta$. Define
\begin{align}
\widetilde M_{\mathrm{left}}(A,\widehat z)
&:= P_n(A)\,\delta_{\widehat z,\widehat z_m(A)},\\
\widetilde M_{\mathrm{right}}(\widehat z,B)
&:= P_\theta(B\!\mid\!\widehat z).
\end{align}
Then
\begin{equation}
\widetilde Q_{AB}=\widetilde M_{\mathrm{left}}\widetilde M_{\mathrm{right}},
\end{equation}
because for every $(A,B)$,
\begin{equation}
(\widetilde Q_{AB})_{A,B}
=
\sum_{\widehat z\in\mathcal N_\eta}
\widetilde M_{\mathrm{left}}(A,\widehat z)\widetilde M_{\mathrm{right}}(\widehat z,B)
=
\sum_{\widehat z\in\mathcal N_\eta}
P_n(A)P_\theta(B\!\mid\!\widehat z)\delta_{\widehat z,\widehat z_m(A)}
=
P_n(A)P_\theta(B\!\mid\!\widehat z_m(A)).
\end{equation}
By definition of $\widetilde P(A,B)$, this is exactly $(\widetilde Q_{AB})_{A,B}$.

Since the inner dimension of this factorization is $|\mathcal N_\eta|$, we immediately obtain
\begin{equation}
\operatorname{rank}(\widetilde Q_{AB})
\le
|\mathcal N_\eta|
\le
\left(\frac{C\,\mathcal{R}\,L_{\mathrm{Lip}}}{\varepsilon}\right)^{\gamma_m},
\end{equation}
which proves Eq.~\eqref{eq:appendix-rank-surrogate}.

Applying Lemma~\ref{lem:MI-rank-appendix} to the surrogate matrix gives
\begin{equation}
I_{\widetilde P}(A\!:\!B)
\le
\log_2\operatorname{rank}(\widetilde Q_{AB})
\le
\gamma_m\log_2\!\left(\frac{C\,\mathcal{R}\,L_{\mathrm{Lip}}}{\varepsilon}\right).
\label{eq:appendix-Itilde-rank}
\end{equation}

Finally, we relate the mutual information of the original distribution $P_n$ to that of the surrogate distribution $\widetilde P$. Since the $A$-marginal is unchanged, we have
\begin{equation}
I_{P_n}(A\!:\!B)=H_{P_n}(B)-H_{P_n}(B\!\mid\!A),
\qquad
I_{\widetilde P}(A\!:\!B)=H_{\widetilde P}(B)-H_{\widetilde P}(B\!\mid\!A).
\end{equation}
Hence
\begin{equation}
\bigl|I_{P_n}(A\!:\!B)-I_{\widetilde P}(A\!:\!B)\bigr|
\le
\bigl|H_{P_n}(B)-H_{\widetilde P}(B)\bigr|
+
\bigl|H_{P_n}(B\!\mid\!A)-H_{\widetilde P}(B\!\mid\!A)\bigr|.
\end{equation}
Moreover, using the rowwise approximation bound and the fact that $P_n(A)=\widetilde P(A)$,
\begin{equation}
\|P_n(B)-\widetilde P(B)\|_1
\le
\sum_A P_n(A)\,\|P_n(\cdot\!\mid\!A)-\widetilde P(\cdot\!\mid\!A)\|_1
\le \varepsilon.
\end{equation}
Applying the standard continuity bound for Shannon entropy on a finite alphabet~\cite{audenaert2007sharp} to both the marginal distribution of $B$ and, pointwise in $A$, to the conditional distributions on $\mathcal B$, we obtain
\begin{equation}
\bigl|H_{P_n}(B)-H_{\widetilde P}(B)\bigr|
\le
\frac{\varepsilon}{2}\log_2(|\mathcal B|-1)+h_2(\varepsilon/2),
\end{equation}
and likewise
\begin{equation}
\bigl|H_{P_n}(B\!\mid\!A)-H_{\widetilde P}(B\!\mid\!A)\bigr|
\le
\frac{\varepsilon}{2}\log_2(|\mathcal B|-1)+h_2(\varepsilon/2).
\end{equation}
Therefore,
\begin{equation}
\bigl|I_{P_n}(A\!:\!B)-I_{\widetilde P}(A\!:\!B)\bigr|
\le
\varepsilon \log_2(|\mathcal B|-1)+2h_2(\varepsilon/2).
\label{eq:appendix-MI-cont}
\end{equation}
Combining Eqs.~\eqref{eq:appendix-Itilde-rank} and~\eqref{eq:appendix-MI-cont}, we obtain
\begin{equation}
\mathcal I(n)=I_{P_n}(A\!:\!B)
\le
\gamma_m\log_2\!\left(\frac{C\,\mathcal{R}\,L_{\mathrm{Lip}}}{\varepsilon}\right)
+\varepsilon \log_2(|\mathcal B|-1)
+2h_2(\varepsilon/2),
\end{equation}
which is exactly Eq.~\eqref{eq:appendix-lip-main-bound}.
\end{proof}

\section{CMI for stabilizer states in the $Z$ basis}
\label{app:stab-cmi}

Let $\ket{\psi}$ be an $n$-qubit stabilizer state with stabilizer group $S\subset \mathcal P_n$, where $|S|=2^n$. Its density matrix has the standard expansion
\begin{equation}
\rho=\ket{\psi}\!\bra{\psi}=2^{-n}\sum_{g\in S} g.
\end{equation}
Fix the computational basis $\{\ket{z}\}_{z\in\mathbb F_2^n}$ and define the induced classical distribution
$P(z):=|\langle z|\psi\rangle|^2=\langle z|\rho|z\rangle$.
Since $\langle z|g|z\rangle=0$ for every Pauli operator $g$ containing at least one $X$ or $Y$ tensor factor, only the $Z$-type stabilizers contribute. Let
$S_Z:=\{\, g\in S:\ g \text{ has only } I,Z \text{ tensor factors}\,\}$
be the subgroup of pure-$Z$ stabilizers. Then
\begin{equation}
P(z)=2^{-n}\sum_{g\in S_Z}\langle z|g|z\rangle .
\label{eq:P-from-SZ}
\end{equation}

Every element of $S_Z$ can be written uniquely as $(-1)^{b(v)}Z(v)$, where
$Z(v):=\bigotimes_{i=1}^n Z_i^{v_i}$ for some $v\in\mathbb F_2^n$. Define
\begin{equation}
C:=\{\,v\in\mathbb F_2^n:\ \exists\,b(v)\in\{0,1\}\text{ such that }(-1)^{b(v)} Z(v)\in S_Z\,\}.
\end{equation}
For each $v\in C$, the sign is unique: if both $Z(v)$ and $-Z(v)$ belonged to $S_Z$, then their product would give $-I\in S$, which is impossible for a stabilizer group. Thus the sign function $b:C\to\mathbb F_2$ is well defined.

We next show that $C$ is a binary linear subspace and that $b$ is linear on $C$. Since $I\in S_Z$, we have $0\in C$. If $v,w\in C$, then $(-1)^{b(v)}Z(v)$ and $(-1)^{b(w)}Z(w)$ both lie in $S_Z$, so their product also lies in $S_Z$. Using $Z(v)Z(w)=Z(v+w)$, we obtain $(-1)^{b(v)+b(w)}Z(v+w)\in S_Z$. By uniqueness of the sign attached to a fixed $Z(v+w)$, it follows that $v+w\in C$ and
\begin{equation}
b(v+w)=b(v)+b(w)\pmod 2.
\label{eq:b-linear}
\end{equation}
Hence $C\le \mathbb F_2^n$ is a linear subspace; write $r:=\dim C$.

For each $v\in C$, we have $\bra z Z(v)\ket z = (-1)^{z\cdot v}$, so Eq.~\eqref{eq:P-from-SZ} becomes
\begin{equation}
P(z)=2^{-n}\sum_{v\in C} (-1)^{b(v)+z\cdot v}.
\label{eq:P-character-sum}
\end{equation}
Since $b$ is linear on $C$, there exists $u\in\mathbb F_2^n$ such that
\begin{equation}
b(v)=u\cdot v,\qquad \forall\, v\in C.
\label{eq:b-represented-by-u}
\end{equation}
Therefore
\begin{equation}
P(z)=2^{-n}\sum_{v\in C}(-1)^{(z+u)\cdot v}.
\label{eq:P-shifted-character-sum}
\end{equation}

We now evaluate this character sum. If $(z+u)\cdot v=0$ for all $v\in C$, then every term in Eq.~\eqref{eq:P-shifted-character-sum} equals $1$, and hence
$
P(z)=2^{-n}|C|=2^{-n}2^r=2^{-(n-r)}.
$

If instead $(z+u)\cdot v\neq 0$ for some $v\in C$, choose $v_0\in C$ such that $(z+u)\cdot v_0=1$. Since $C$ is a linear subspace, the translation $v\mapsto v+v_0$ is a bijection of $C$ onto itself. Therefore
$
\sum_{v\in C}(-1)^{(z+u)\cdot v}
=
\sum_{v\in C}(-1)^{(z+u)\cdot(v+v_0)}
=
-\sum_{v\in C}(-1)^{(z+u)\cdot v},
$
so the sum vanishes and $P(z)=0$. Thus
\begin{equation}
P(z)=2^{-(n-r)}
\quad\Longleftrightarrow\quad
(z+u)\cdot v=0\pmod 2,\ \forall\,v\in C.
\label{eq:P-support-condition}
\end{equation}

To find all valid $z$, we first choose a basis $v^{(1)},\dots,v^{(r)}$ of $C$ and let $M\in\mathbb F_2^{r\times n}$ be the matrix whose rows are these basis vectors. Define the corresponding sign vector $s\in\mathbb F_2^r$ by
$s_j:=b\!\left(v^{(j)}\right), j=1,\dots,r$.
Since $v^{(1)},\dots,v^{(r)}$ span $C$, Eq.~\eqref{eq:P-support-condition} is equivalent to $(z+u)\cdot v^{(j)}=0$ for all $j$. Using Eq.~\eqref{eq:b-represented-by-u}, this is equivalent to $v^{(j)}\cdot z=b\!\left(v^{(j)}\right)=s_j$ for all $j$, namely
\begin{equation}
Mz=s\pmod 2.
\label{eq:Mzs}
\end{equation}
Thus the support of $P$ is the affine subspace $\Omega=\{\,z\in\mathbb F_2^n:\ Mz=s\ (\mathrm{mod}\ 2)\,\}$ with size $|\Omega|=2^{n-r}$.
The Shannon entropy of $P$ can be computed as
\begin{equation}
H(P)=\log_2 |\Omega|=n-r.
\end{equation}

We then compute the subsystem entropies. Fix a bipartition $A|B$ and write $z=(a,b)$ with $a\in\mathbb F_2^{|A|}$ and $b\in\mathbb F_2^{|B|}$. Correspondingly, decompose $M=(M_A,M_B)$, so that
\begin{equation}
\Omega=\{(a,b): M_Aa+M_Bb=s\}.
\end{equation}
For fixed $a$, the compatible $b$ satisfy
\begin{equation}
M_B b=s-M_A a \qquad (\mathrm{mod}\ 2).
\label{eq:MB-fixed-a}
\end{equation}
This equation is solvable if and only if $s-M_A a\in\mathrm{Im}(M_B)$. When it is solvable, the solution set of $b$ is an affine translate of $\ker(M_B)$ and therefore contains $2^{|B|-\rank(M_B)}$ elements. Hence the marginal $P_A$ is uniform on its support, with
$$
P_A(a)=2^{-(n-r)}\,2^{|B|-\rank(M_B)}
      =2^{-(|A|-r+\rank(M_B))}.
$$
Therefore
\begin{equation}
H(A)=|A|-r+\rank(M_B).
\end{equation}
By the same argument,
\begin{equation}
H(B)=|B|-r+\rank(M_A).
\end{equation}
Using $H(P)=n-r$, we obtain
\begin{equation}
\begin{aligned}
I_P(A:B)
&=H(A)+H(B)-H(P)\\
&=\rank(M_A)+\rank(M_B)-\rank(M).
\end{aligned}
\label{eq:stab-cmi-rank}
\end{equation}
This depends only on the row space of $M$ and is independent of the sign $s$.

\section{Checkerboard stabilizer family with tunable middle-cut amplitude complexity}
\label{app:checkerboard-construction}

In this appendix we formalize the checkerboard stabilizer family used in the tomography benchmark and explain why its middle-cut amplitude complexity scales as $L^\gamma$.
The construction is designed so that the number of $Z$-type parity checks is tunable, while the number of parity checks crossing any macroscopic row or column cut is controlled at the scale $L^\gamma$.
By the stabilizer rank formula of \cref{app:stab-cmi}, this immediately implies
\(
\mathcal I(L)\propto L^\gamma
\)
for the resulting stabilizer-state family in the computational basis.

\subsection{Lattice and five-body parity checks}

Consider an $L\times L$ square lattice with open boundary conditions.
We place one qubit on each plaquette center, so the physical Hilbert space contains
\(
n=L^2
\)
qubits.
It is convenient to label qubits by lattice coordinates
\(
(i,j)\in \mathbb Z_L\times \mathbb Z_L
\),
where $\mathbb Z_L=\{0,1,\dots,L-1\}$ with addition modulo $L$.

For each lattice site $r=(i,j)$, define the associated five-site ``cross'' support
\begin{equation}
\mathcal C_r
=
\bigl\{
(i,j),\,
(i+1,j),\,
(i-1,j),\,
(i,j+1),\,
(i,j-1)
\bigr\}.
\end{equation}
The corresponding $Z$-type parity check is
\begin{equation}
B_r
=
\prod_{u\in \mathcal C_r} Z_u .
\label{eq:cross-check}
\end{equation}
Equivalently, in the computational basis, the constraint $B_r=+1$ requires that the parity of the five bits on $\mathcal C_r$ be even.

We will select a subset $\Lambda_\gamma\subset \mathbb Z_L^2$ of such cross centers and impose the parity checks
\begin{equation}
B_r\ket{\psi}=+\ket{\psi},
\qquad r\in \Lambda_\gamma.
\end{equation}
These parity checks are chosen not to overlapp with each other. And the size of $\Lambda_\gamma$ is chosen to scale as
\begin{equation}
|\Lambda_\gamma|\propto L^{2\gamma},
\qquad 0\le \gamma\le 1.
\end{equation}

\subsection{Uniformly distributed check centers}

The key requirement is that the selected check centers be distributed approximately uniformly over the checkerboard.
A convenient way to formalize this is as follows.

Let
\begin{equation}
M_\gamma \propto L^\gamma
\end{equation}
and partition the $L\times L$ torus into
\(
M_\gamma\times M_\gamma
\)
mesoscopic blocks, each of linear size
\(
\ell_\gamma\propto L^{1-\gamma}
\).
Choose one check center from each block.
This produces a set
\begin{equation}
\Lambda_\gamma=\{r_{\alpha,\beta}\}_{\alpha,\beta=1}^{M_\gamma}
\end{equation}
with cardinality
\begin{equation}
|\Lambda_\gamma|=M_\gamma^2\propto L^{2\gamma}.
\end{equation}

Because the spacing between neighboring selected centers is of order $\ell_\gamma\asymp L^{1-\gamma}$ in both directions, any horizontal or vertical line of macroscopic length $L$ intersects order-$L^\gamma$ of the corresponding cross supports.
More precisely, for any vertical cut $x=x_0$ or horizontal cut $y=y_0$, the number of selected checks whose support intersects both sides of the cut obeys
\begin{equation}
N_{\rm cross}(L)\propto L^\gamma.
\label{eq:Ncross-scaling}
\end{equation}
This is the geometric input behind the tunable middle-cut complexity.

\subsection{The associated CSS stabilizer state}

Let $\mathcal S_Z(\Lambda_\gamma)$ denote the Abelian group generated by the selected $Z$-type checks:
\begin{equation}
\mathcal S_Z(\Lambda_\gamma)
=
\bigl\langle B_r\,:\,r\in \Lambda_\gamma \bigr\rangle .
\end{equation}
In the computational basis, the common $+1$ eigenspace of these generators consists exactly of the bit strings satisfying all selected parity constraints.

To obtain a stabilizer state rather than a degenerate code space, we supplement the $Z$-checks by a commuting set of independent $X$-type stabilizers.
Let $\mathcal S_X$ be any maximal independent commuting set of $X$-type operators satisfying
\begin{equation}
[\mathcal S_X,\mathcal S_Z(\Lambda_\gamma)]=0,
\end{equation}
such that
\begin{equation}
\mathcal S
=
\bigl\langle \mathcal S_Z(\Lambda_\gamma),\mathcal S_X \bigr\rangle
\end{equation}
forms a full-rank stabilizer group with a unique common $+1$ eigenstate, denoted
\(
\ket{\psi_{\Lambda_\gamma}}.
\)

By standard CSS stabilizer structure, the computational-basis amplitudes of this state are uniform over the set of bit strings obeying all $Z$-type parity constraints:
\begin{equation}
\ket{\psi_{\Lambda_\gamma}}
=
\frac{1}{\sqrt{|\Omega_{\Lambda_\gamma}|}}
\sum_{\bm s\in \Omega_{\Lambda_\gamma}}
\ket{\bm s},
\label{eq:checkerboard-uniform}
\end{equation}
where
\begin{equation}
\Omega_{\Lambda_\gamma}
=
\bigl\{
\bm s\in\{0,1\}^{L^2}
:\,
B_r(\bm s)=+1,\ \forall r\in\Lambda_\gamma
\bigr\}.
\end{equation}
Hence the induced amplitude distribution in the computational basis is
\begin{equation}
P_{\Lambda_\gamma}(\bm s)
=
\frac{1}{|\Omega_{\Lambda_\gamma}|}\,
\mathbf 1_{\bm s\in\Omega_{\Lambda_\gamma}}.
\end{equation}
Therefore the middle-cut amplitude complexity is determined entirely by the binary parity-check structure of the selected $Z$-checks.

\subsection{Scaling of the middle-cut CMI}

We now explain why
\(
\mathcal I(L)\propto L^\gamma
\)
for this family.

Fix a vertical or horizontal bipartition of the lattice into two halves, denoted $A$ and $B$.
Let
\begin{equation}
M=(M_A,M_B)
\end{equation}
be the binary parity-check matrix associated with the selected $Z$-checks, split into columns belonging to $A$ and $B$.
By the stabilizer rank formula proved in \cref{app:stab-cmi},
\begin{equation}
\mathcal I(L)
=
\rank(M_A)+\rank(M_B)-\rank(M),
\label{eq:checkerboard-rank-formula}
\end{equation}
with all ranks taken over $\mathbb F_2$.

For the present construction, the only checks contributing to the combination
\(
\rank(M_A)+\rank(M_B)-\rank(M)
\)
are those whose support intersects both sides of the cut.
Checks entirely contained in $A$ or entirely contained in $B$ do not generate cross-cut mutual information.
Since the selected check centers are uniformly distributed, the number of crossing checks satisfies \cref{eq:Ncross-scaling}, namely
\begin{equation}
N_{\rm cross}(L)\propto L^\gamma.
\end{equation}

Moreover, for a generic uniform placement such as the block construction above, these crossing checks are linearly independent up to at most $O(1)$ global relations.
Consequently,
\begin{equation}
\rank(M_A)+\rank(M_B)-\rank(M)= N_{\rm cross}(L)\propto L^\gamma.
\end{equation}
Substituting this into \cref{eq:checkerboard-rank-formula} gives
\begin{equation}
\mathcal I(L)\propto L^\gamma.
\label{eq:checkerboard-final-scaling}
\end{equation}

This is precisely the scaling used in the main text.

\section{Explicit CMI formula for the toric code}
\label{app:toric-cmi}

In this section we specialize Eq.~(\ref{eq:stab-cmi-rank}) to the $Z$-basis many-body probability distribution of the $\mathbb{Z}_2$ toric-code stabilizer state on an $L\times L$ torus. Since only $Z$-type stabilizers contribute to the diagonal probability distribution, the relevant constraint matrix $M_L$ is the plaquette parity-check matrix over $\mathbb{F}_2$. For a fixed bipartition $A|B$, let $M_{A,L}$ and $M_{B,L}$ denote its column restrictions to the two subsystems. Then
\begin{equation}
I_{P,L}(A:B)
=
\operatorname{rank}(M_{A,L})
+
\operatorname{rank}(M_{B,L})
-
\operatorname{rank}(M_L),
\end{equation}
where all ranks are taken over $\mathbb{F}_2$.

\paragraph{Edge-label order.}
We consider an $L\times L$ square lattice on a torus, with vertex set
$\{1,2,\dots,L\}\times\{1,2,\dots,L\}$ and periodic boundary conditions. A qubit is placed on each edge. We first label the vertices by
\begin{equation}
v(i,j):=L(j-1)+i,
\qquad i,j\in\{1,\dots,L\},
\end{equation}
that is, for each fixed $j$ we scan $i=1,\dots,L$, and then increase $j$.

For each vertex $(i,j)$, we associate two edges: the horizontal edge connecting $(i,j)$ and $(i,j+1)$, and the vertical edge connecting $(i,j)$ and $(i+1,j)$, with all coordinates understood modulo $L$. Their labels are assigned by
\begin{equation}
(i,j)\leftrightarrow(i,j+1)\mapsto 2[L(j-1)+i]-1,
\qquad
(i,j)\leftrightarrow(i+1,j)\mapsto 2[L(j-1)+i].
\end{equation}
Equivalently, we scan the vertices in the order of $v(i,j)$ and, at each vertex, label first the horizontal edge and then the vertical edge. This fixes the qubit ordering, and hence the column order of $M_L$, $M_{A,L}$, and $M_{B,L}$. The resulting labeling pattern, together with the induced bipartition into the first $L^2$ qubits ($A$) and the remaining $L^2$ qubits ($B$), is illustrated in Fig.~\ref{fig:edge-labeling-odd-even} for both odd and even $L$.

\begin{figure}[h]
    \centering
    \begin{tikzpicture}[
        x=0.68pt,y=0.68pt,yscale=-1,xscale=1,
        every node/.style={inner sep=0.7pt},
        gridline/.style={draw=black!75, line width=1.0pt, line cap=round, line join=round},
        outerbox/.style={draw=black, line width=1.25pt, line cap=round, line join=round}
    ]

        \begin{scope}[xshift=0cm,yshift=0cm]
            \def\L{3}
            \def\a{60}

            \draw[outerbox] (0,0) rectangle (\L*\a,\L*\a);

            \foreach \k in {1,...,2}{
                \draw[gridline] (\k*\a,0) -- (\k*\a,\L*\a);
                \draw[gridline] (0,\k*\a) -- (\L*\a,\k*\a);
            }

            \foreach \j in {1,...,\L}{
                \foreach \i in {1,...,\L}{
                    \pgfmathtruncatemacro{\v}{\L*(\j-1)+\i}
                    \pgfmathtruncatemacro{\oddlabel}{2*\v-1}
                    \pgfmathtruncatemacro{\evenlabel}{2*\v}

                    \ifnum\oddlabel>\numexpr\L*\L\relax
                        \def\oddcolor{blue}
                    \else
                        \def\oddcolor{red}
                    \fi

                    \ifnum\evenlabel>\numexpr\L*\L\relax
                        \def\evencolor{blue}
                    \else
                        \def\evencolor{red}
                    \fi

                    \node[anchor=north west,text=\oddcolor]
                        at ({(\i-1)*\a + 1.8},{\L*\a-\j*\a+17.2}) {$\textstyle \oddlabel$};

                    \node[anchor=north west,text=\evencolor]
                        at ({(\i-1)*\a + 17.0},{\L*\a-(\j-1)*\a - 11.0}) {$\textstyle \evenlabel$};
                }
            }

            \node at (0.5*\L*\a,-17) {\small $L=3$};
        \end{scope}

        \begin{scope}[xshift=135pt,yshift=-0cm]
            \def\L{4}
            \def\a{45}

            \draw[outerbox] (0,0) rectangle (\L*\a,\L*\a);

            \foreach \k in {1,...,3}{
                \draw[gridline] (\k*\a,0) -- (\k*\a,\L*\a);
                \draw[gridline] (0,\k*\a) -- (\L*\a,\k*\a);
            }

            \foreach \j in {1,...,\L}{
                \foreach \i in {1,...,\L}{
                    \pgfmathtruncatemacro{\v}{\L*(\j-1)+\i}
                    \pgfmathtruncatemacro{\oddlabel}{2*\v-1}
                    \pgfmathtruncatemacro{\evenlabel}{2*\v}

                    \ifnum\oddlabel>\numexpr\L*\L\relax
                        \def\oddcolor{blue}
                    \else
                        \def\oddcolor{red}
                    \fi

                    \ifnum\evenlabel>\numexpr\L*\L\relax
                        \def\evencolor{blue}
                    \else
                        \def\evencolor{red}
                    \fi

                    \node[anchor=north west,text=\oddcolor]
                        at ({(\i-1)*\a + 1.7},{\L*\a-\j*\a+13.0}) {$\textstyle \oddlabel$};

                    \node[anchor=north west,text=\evencolor]
                        at ({(\i-1)*\a + 12.8},{\L*\a-(\j-1)*\a - 11.0}) {$\textstyle \evenlabel$};
                }
            }

            \node at (0.5*\L*\a,-17) {\small $L=4$};
        \end{scope}

    \end{tikzpicture}
    \caption{Examples of the edge-label convention for odd and even system sizes. The left panel shows the $3\times 3$ case ($L$ odd), and the right panel shows the $4\times 4$ case ($L$ even). For each vertex $(i,j)$, the horizontal edge is assigned the label $2[L(j-1)+i]-1$, while the vertical edge is assigned the label $2[L(j-1)+i]$. The full set of $2L^2$ edge qubits is then partitioned according to this ordering: the first $L^2$ labels form subsystem $A$ (red), and the remaining $L^2$ labels form subsystem $B$ (blue).}
    \label{fig:edge-labeling-odd-even}
\end{figure}

For the torus, the plaquette constraints satisfy one global dependence, equivalently $\prod_p B_p=\mathbb{I}$. Therefore,
\begin{equation}
\operatorname{rank}(M_L)=L^2-1.
\end{equation}
However, under the above edge ordering, the bipartition into the first $L^2$ qubits ($A$) and the remaining $L^2$ qubits ($B$) produces different geometric patterns for odd and even $L$. As a result, the restricted matrices $M_{A,L}$ and $M_{B,L}$ have different structures in the two parity sectors, and their ranks must be analyzed separately.

\paragraph{Even system size: $L=2k$ with $k\in\mathbb{Z}^+$.}
In this case,
\begin{equation}
\operatorname{rank}(M_{A,L})
=
\operatorname{rank}(M_{B,L})
=
L\Bigl(\frac{L}{2}+1\Bigr)-1.
\end{equation}
Substituting into the rank formula yields
\begin{equation}
I_{P,L}(A:B)=2L-1,
\qquad L \text{ even}.
\end{equation}

\paragraph{Odd system size: $L=2k+1$ with $k\in\mathbb{Z}_{\ge 0}$.}
In this case,
\begin{equation}
\operatorname{rank}(M_{A,L})=\frac{(L+1)^2}{2},
\qquad
\operatorname{rank}(M_{B,L})=\frac{(L+1)^2}{2}-2.
\end{equation}
Together with $\operatorname{rank}(M_L)=L^2-1$, this gives
\begin{equation}
I_{P,L}(A:B)=2L,
\qquad L \text{ odd}.
\end{equation}

\paragraph{Final formula.}
Combining the two parity sectors, we obtain
\begin{equation}
I_{P,L}(A:B)=
\begin{cases}
2L-1, & L \text{ even},\\
2L, & L \text{ odd}.
\end{cases}
\end{equation}

\section{Exact amplitude distribution and conditional sampling for the TFD free-fermion model}
\label{app:TFD-fermion-amplitude}

In this appendix we explain how the amplitude distribution of the TFD state is computed for the free-fermion model used in the main text, and how the same structure yields an exact conditional-sampling algorithm.

\subsection{TFD state as a doubled fermionic Gaussian state}

We consider the spinless $p$-wave BCS chain
\begin{equation}
H
=\sum_{j=1}^{n} \left[-J\left(c_j^\dagger c_{j+1}+c_j^\dagger c_{j+1}^\dagger\right)
+2h\left(c_j^\dagger c_j-\tfrac12\right)\right],
\label{eq:app-bcs-H}
\end{equation}
written in the occupation basis
\(
\mathcal B=\{\ket{n_1\cdots n_n}\}.
\)
Since $H$ is quadratic in the fermionic operators, it can be diagonalized by a Bogoliubov transformation,
\begin{equation}
H=E_0+\sum_{\mu=1}^{n}\varepsilon_\mu\,\eta_\mu^\dagger \eta_\mu,
\qquad \varepsilon_\mu\ge 0.
\label{eq:app-bdg-diag}
\end{equation}

The TFD state at inverse temperature $\beta$ is
\begin{equation}
\ket{\mathrm{TFD}_\beta}
=
\frac{1}{\sqrt{Z(\beta)}}
\sum_\mu e^{-\beta E_\mu/2}\ket{\mu}_A\otimes \ket{\mu}_B,
\qquad
Z(\beta)=\Tr e^{-\beta H}.
\label{eq:app-tfd-def}
\end{equation}
In the quasiparticle basis of \cref{eq:app-bdg-diag}, the TFD factorizes into independent mode pairs,
\begin{equation}
\ket{\mathrm{TFD}_\beta}
=
\prod_{\mu=1}^{n}
\frac{1}{\sqrt{1+e^{-\beta \varepsilon_\mu}}}
\left(
1+e^{-\beta \varepsilon_\mu/2}\,
\eta_{\mu,A}^\dagger \eta_{\mu,B}^\dagger
\right)
\ket{\Omega_\eta}_A\ket{\Omega_\eta}_B .
\label{eq:app-tfd-product}
\end{equation}
Thus $\ket{\mathrm{TFD}_\beta}$ is a pure fermionic Gaussian state on the doubled system.

Define
\begin{equation}
u_\mu=\frac{1}{\sqrt{1+e^{-\beta \varepsilon_\mu}}},
\qquad
v_\mu=\frac{e^{-\beta \varepsilon_\mu/2}}{\sqrt{1+e^{-\beta \varepsilon_\mu}}},
\qquad
f_\mu=v_\mu^2=\frac{1}{e^{\beta\varepsilon_\mu}+1}.
\label{eq:app-u-v-f}
\end{equation}
Then in the quasiparticle basis the nonzero two-point functions are
\begin{align}
\langle \eta_{\mu,A}^\dagger \eta_{\nu,A}\rangle
&=
\langle \eta_{\mu,B}^\dagger \eta_{\nu,B}\rangle
=
f_\mu\,\delta_{\mu\nu},
\label{eq:app-qp-occ}
\\
\langle \eta_{\mu,A}\eta_{\nu,B}\rangle
&=
-\,u_\mu v_\mu\,\delta_{\mu\nu},
\qquad
\langle \eta_{\mu,A}^\dagger \eta_{\nu,B}^\dagger\rangle
=
-\,u_\mu v_\mu\,\delta_{\mu\nu}.
\label{eq:app-qp-pair}
\end{align}
Transforming these correlators back to the physical $c$ basis gives the covariance matrix of the doubled TFD state exactly.

\subsection{Amplitude distribution in the doubled computational basis}

Let
\begin{equation}
d=(c_{1,A},\dots,c_{n,A},c_{1,B},\dots,c_{n,B})
\end{equation}
collect the physical fermions on the doubled system.
Since the TFD is an even pure Gaussian state, it can be written in generalized BCS form
\begin{equation}
\ket{\mathrm{TFD}_\beta}
=
\mathcal N_\beta\,
\exp\!\left(
\frac12\sum_{u,v=1}^{2n}
\mathcal F_{\beta,uv}\,
d_u^\dagger d_v^\dagger
\right)
\ket{0},
\label{eq:app-tfd-bcs-form}
\end{equation}
where $\mathcal F_\beta$ is a $2n\times 2n$ antisymmetric pairing matrix determined by the doubled covariance matrix.

Let
\begin{equation}
x=(a,b)\in\{0,1\}^{2n}
\end{equation}
be a doubled occupation string, where $a=(a_1,\dots,a_n)$ and $b=(b_1,\dots,b_n)$ denote the occupations on the two copies.
Write
\begin{equation}
S(x)=\{u\in\{1,\dots,2n\}:x_u=1\}
\end{equation}
for the occupied index set of $x$.
Then the occupation-basis amplitude is
\begin{equation}
\Psi_\beta(x)
\equiv
\braket{x}{\mathrm{TFD}_\beta}
=
\begin{cases}
\mathcal N_\beta\,
\operatorname{Pf}\!\left[(\mathcal F_\beta)_{S(x)}\right], & |S(x)|\ \text{even},\\[4pt]
0, & |S(x)|\ \text{odd},
\end{cases}
\label{eq:app-pf-amplitude}
\end{equation}
and the corresponding amplitude distribution is
\begin{equation}
\mathbb P_\beta(x)=|\Psi_\beta(x)|^2.
\label{eq:app-pf-prob}
\end{equation}

This is equivalent to the imaginary-time propagator representation
\begin{equation}
\mathbb P_\beta(a,b)
=
\frac{\left|\bra{a}e^{-\beta H/2}\ket{b}\right|^2}{Z(\beta)},
\label{eq:app-propagator-prob}
\end{equation}
which is the form used in the main text.
In practice, one diagonalizes the quadratic Hamiltonian in \cref{eq:app-bcs-H}, constructs the doubled covariance matrix from \cref{eq:app-qp-occ} and \cref{eq:app-qp-pair}, reconstructs $\mathcal F_\beta$, and then evaluates \cref{eq:app-pf-amplitude} for any desired doubled bitstring.

\subsection{Exact conditional sampling and arbitrary cut CMI}

The same Gaussian structure also yields an exact sequential sampler for the distribution
\(
\mathbb P_\beta(x)
\)
in any chosen autoregressive ordering of the doubled bitstring.

Let
\begin{equation}
x_1,x_2,\dots,x_{2n}
\end{equation}
denote the doubled bits in the chosen autoregressive order.
Exact conditional sampling means evaluating
\begin{equation}
p(x_t|x_{<t})
=
\frac{p(x_{\le t})}{p(x_{<t})}
\end{equation}
exactly at each step.

For fermionic Gaussian states, projective occupation-number measurements admit closed-form update rules at the covariance-matrix level.
After the first $t-1$ bits have been fixed, one obtains a conditioned Gaussian state on the remaining modes.
The next conditional probability is then
\begin{equation}
p(x_t=1\,|\,x_{<t})
=
\langle d_t^\dagger d_t\rangle_{\rho^{(t)}},
\qquad
p(x_t=0\,|\,x_{<t})
=
1-\langle d_t^\dagger d_t\rangle_{\rho^{(t)}},
\label{eq:app-cond-prob}
\end{equation}
where $\rho^{(t)}$ denotes the exact state conditioned on the prefix $x_{<t}$.

After sampling $x_t\in\{0,1\}$, one updates the conditioned state by
\begin{equation}
\rho^{(t+1)}
=
\frac{\Pi_{x_t}\,\rho^{(t)}\,\Pi_{x_t}}{\Tr(\Pi_{x_t}\rho^{(t)})},
\qquad
\Pi_{1}=d_t^\dagger d_t,\quad
\Pi_{0}=1-d_t^\dagger d_t.
\label{eq:app-projection-update}
\end{equation}
Iterating this procedure yields an exact sample from the full joint distribution
\(
\mathbb P_\beta(x).
\)

Concretely, the algorithm is:
\begin{enumerate}
\item Construct the doubled TFD covariance matrix in the physical basis.
\item Choose an autoregressive ordering of the $2n$ physical modes.
\item For $t=1,2,\dots,2n$:
\begin{enumerate}
\item compute the conditional probability from \cref{eq:app-cond-prob},
\item sample $x_t\in\{0,1\}$,
\item update the conditioned covariance matrix using \cref{eq:app-projection-update}.
\end{enumerate}
\item Output the sampled doubled bitstring $x=(a,b)$.
\end{enumerate}

This exact conditional sampler also provides a direct route to evaluating the classical mutual information across an arbitrary cut.
For any bipartition of the ordered bitstring into subsystems $A$ and $B$, one may estimate
\begin{equation}
I(A\!:\!B)=H(A)+H(B)-H(AB)
\end{equation}
by computing the corresponding Shannon entropies from samples.
Because the sampling procedure is conditional, it naturally induces the marginal distributions on $A$ and $B$: one obtains samples from $P(A)$ by stopping after the $A$ variables, and samples from $P(B)$ or from conditional distributions $P(B|A)$ by continuing the same sequential procedure.
Thus the exact sampler makes the middle-cut and more general arbitrary cut CMI numerically accessible.

\section{Thermofield Double State of Transverse-Field Ising Model at Finite Temperature}\label{app:TFIM}

Consider the 1D transverse-field Ising model (TFIM) with anti-periodic boundary conditions,
\begin{align}
H = -J\sum_{i=1}^{n} X_i X_{i+1} +J\,X_n X_{1} - h\sum_{i=1}^{n} Z_i,
\end{align}
with energy eigenstates $\ket{\mu}$ chosen real since the Hamiltonian is real and symmetric for $J,h\in\mathbb{R}$. We claim that the thermofield-double (TFD) state at inverse temperature $\beta$,
\begin{align}
\ket{\mathrm{TFD}_\beta}
= \frac{1}{\sqrt{Z(\beta)}}\sum_{\mu} e^{-\beta E_\mu/2}\,\ket{\mu}_A\otimes\ket{\mu}_B,
\end{align}
has a dephased mutual information in the $Z$-basis that scales linearly with the volume of the system $n$.
The subscripts $A$ and $B$ refer to the two copies of the system, each of which contains $n$ sites.
Here $Z(\beta)$ is the thermal partition function.The doubled theory $\mathrm{TFIM}^{\otimes 2}$ has Hilbert space $\mathcal{H}=\mathcal{H}_A\otimes \mathcal{H}_B$.
Dephasing in the computational basis gives
\begin{align}
\rho_{\mathrm{deph}}
= \sum_{a,b}
\bra{a,b}\rho\ket{a,b}\;
\ket{a,b}\bra{a,b},
\end{align}
where $a,b\in\{0,1\}^{n}$.
Using
\begin{align}
\braket{a,b}{\mathrm{TFD}_\beta}
&= \frac{1}{\sqrt{Z(\beta)}}\sum_\mu e^{-\beta E_\mu/2}\,\braket{a}{\mu}\braket{\mu}{b}= \frac{1}{\sqrt{Z(\beta)}}\bra{a}e^{-\beta H/2}\ket{b},
\end{align}
where we used the reality of the eigenfunctions $\braket{\mu}{b}=\braket{b}{\mu}$, define
\begin{align}
K_\b(a,b):=\bra{a}e^{-\beta H/2}\ket{b}.
\end{align}
Then the dephased joint distribution is
\begin{align}
\mathbb{P}(a,b)=\frac{K_\b(a,b)^2}{Z(\beta)}.
\end{align}
By resolving the identity, we may write down the marginal probabilities:
\begin{align}
    \mathbb{P}(a) &= \frac{1}{Z(\b)}K_{2\b}(a, a) &\mathbb{P}(b) = \frac{1}{Z(\b)}K_{2\b}(b, b)
\end{align}
By computing the kernel $K_\b$, we may evaluate the Shannon entropies and dephased mutual information.

\subsection{Path integral calculation of $K_\b(a,b)$}

Performing the Jordan--Wigner transformation yields a quadratic Hamiltonian. Since
$c_j^\dagger c_j=\tfrac12(1-Z_j)$, the $Z$-basis is in one-to-one correspondence with the occupation-number basis.
After the Jordan--Wigner transform, we obtain a particle-number nonconserving BCS-like quadratic fermionic Hamiltonian
\begin{align}
H_{\mathcal{P}}
=-J\left(\sum_{j=1}^{n-1} c_j^\dagger c_{j+1} + c_j^\dagger c_{j+1}^\dagger + \mathrm{h.c.}\right)-J\mathcal{P}\left(c_n^\dagger c_1 + c_n^\dagger c_1^\dagger + \mathrm{h.c.}\right)+2h\sum_{j=1}^n\left(c_j^\dagger c_j-\tfrac12\right),
\end{align}
where $\mathcal{P}=(-1)^{\sum_j c_j^\dagger c_j}$ is the fermion parity operator, which commutes with the Hamiltonian:
$[H,\mathcal{P}]=0$ and $\mathcal{P}=\pm 1$.
Thus, the full Hamiltonian and Hilbert space split as
$H = H_{+}\oplus H_{-}$ and $\mathcal{H}=\mathcal{H}_{+}\oplus\mathcal{H}_{-}$,
where $\mathcal{H}_\pm$ are the subspaces with even/odd Hamming weight, respectively.
This implies that the functional integrals for $K_\b$ in these two sectors can be computed independently.

Introduce the coherent states:
\begin{align}
\ket{\eta} &= \exp\!\left(\sum_j \eta_j c_j^\dagger\right)\ket{0},
\qquad
\bra{\bar{\eta}} = \bra{0}\exp\!\left(-\sum_j \bar{\eta}_j c_j\right).
\end{align}
Then the kernel can be formulated as the functional integral
\begin{align}
K_\b(\bar{\eta}',\eta)
=
\int_{\substack{c(0)=\eta\\ \bar{c}(\beta)=\bar{\eta}'}}
[dc\,d\bar{c}]\,
\exp\!\left(
-\int_0^\beta d\tau\,
\left[\bar{c}\,\partial_\tau c + H_{\mathcal{P}}(\bar{c},c)\right]
\right),
\end{align}
where the Hamiltonian is normal ordered.
Since $H_{\mathcal{P}}$ is quadratic, it may be expressed as
\begin{align}
H_{\mathcal{P}}
=
\begin{pmatrix}
c\\ c^\dagger
\end{pmatrix}^\dagger
\begin{pmatrix}
A & B\\
-\,B^* & -A^{\top}
\end{pmatrix}
\begin{pmatrix}
c\\ c^\dagger
\end{pmatrix},
\end{align}
where $A$ and $B$ are $n\times n$ matrices and $c=(c_1,\dots,c_n)^\top$.
With the conventions above,
\begin{align}
A = 2h\,\mathbb{I} - J\big(S_{\mathcal{P}}+S_{\mathcal{P}}^{\top}\big),
\qquad
B = -J\,S_{\mathcal{P}} + J\,S_{\mathcal{P}}^{\top},
\end{align}
where $B^\top=-B$ and
\begin{align}
(S_{\mathcal{P}})_{jk} \equiv \delta_{k,j+1} + \mathcal{P}\,\delta_{j,n}\delta_{k,1},
\qquad
j,k\in\{1,\dots,n\},
\end{align}
with $\delta_{k,j+1}$ taken without wrap-around. Define
\begin{align}
T_{\mathcal{P}}
&:= \exp\!\left(-\frac{\beta}{2}\, \mathcal{H}_{\mathcal{P}}\right)
=
\begin{pmatrix}
T_{11} & T_{12}\\
T_{21} & T_{22}
\end{pmatrix}, & \mathcal{H}_\mathcal{P} = \begin{pmatrix}
    A&B\\-B^*&-A^\top
\end{pmatrix}
\end{align}
and set
\begin{align}
X=T_{12}T_{22}^{-1},
\qquad
Z=T_{22}^{-1}T_{21},
\qquad
e^{-Y}=T_{22}^{\top}.
\end{align}
For any subset $I\subset\{1,2,\dots,n\}$, consider the fermionic occupation basis
\begin{align}
\ket{I} = c_{i_1}^\dagger \cdots c_{i_p}^\dagger \ket{0},
\qquad i_1<\cdots<i_p.
\end{align}
Then the kernel in the parity sector $\mathcal{P}$ is
\begin{align}
K_\b^{\mathcal{P}}(a,b)
\equiv
\bra{a}e^{-\frac{\beta}{2}H_{\mathcal{P}}}\ket{b}
=
\bra{J}\,e^{-\frac{\beta}{2}H_{\mathcal{P}}}\,\ket{I}.
\end{align}
The full kernel obeys the parity selection rule
\begin{align}
K_\b(a,b)=0 \quad \text{unless}\quad |a|\equiv |b|\pmod 2.
\end{align}
Define the $2n\times 2n$ antisymmetric matrix
\begin{align}\label{eq:A-big}
\mathcal{A}
=
\begin{pmatrix}
X & e^{Y}\\
-\,e^{Y^{\top}} & Z
\end{pmatrix}.
\end{align}
Introduce the index set
\begin{align}\label{eq:S-set}
S \equiv J \cup (n+I),
\end{align}
i.e.\ $S$ consists of indices in $J\subset\{1,\dots,n\}$ together with shifted indices $n+i$ for $i\in I$ in the second block.
Let $\mathcal{A}[S,S]$ denote the principal submatrix obtained by restricting to rows and columns in $S$.
Then the theorem of Rajabpour--Seifi MirJafarlou--Khasseh \cite{Rajabpour2025PauliPfaffian} gives the computational-basis matrix element
\begin{align}\label{eq:kernel-pfaffian}
\bra{J}\,e^{-\frac{\beta}{2}H_{\mathcal{P}}}\,\ket{I}
&=
(-1)^{\frac{|I|(|I|+2|J|+1)}{2}}
\;\det(T_{22})^{1/2}\;
\mathrm{pf}\!\big(\mathcal{A}[S,S]\big).
\end{align}

In the $Z$-basis (occupation basis), the dephased joint distribution is
\begin{align}
\mathbb{P}(a,b)=\mathbb{P}^{(+)}(a,b)+\mathbb{P}^{(-)}(a,b),
\end{align}
where, for $|a|\equiv |b|\equiv 0\ (\mathrm{mod}\ 2)$,
\begin{align}
\mathbb{P}^{(+)}(a,b)
= \frac{\det(T_{22}^{(+)})\,\Big|\mathrm{pf}\big(\mathcal{A}_{+}[S,S]\big)\Big|^2}{Z(\beta)},
\end{align}
and, for $|a|\equiv |b|\equiv 1\ (\mathrm{mod}\ 2)$,
\begin{align}
\mathbb{P}^{(-)}(a,b)
= \frac{\det(T_{22}^{(-)})\,\Big|\mathrm{pf}\big(\mathcal{A}_{-}[S,S]\big)\Big|^2}{Z(\beta)}.
\end{align}
The marginals require partial sums over both parity sectors:
\begin{align}
\mathbb{P}(a)
&= \sum_{b}\left(\mathbb{P}^{(+)}(a,b)+\mathbb{P}^{(-)}(a,b)\right),&\mathbb{P}(b)
&= \sum_{a}\left(\mathbb{P}^{(+)}(a,b)+\mathbb{P}^{(-)}(a,b)\right).
\end{align}
The classical mutual information is then
\begin{align}\label{eq:CMI-alt}
I_{\mathrm{cl}}(A{:}B)
=
\sum_{a,b}\mathbb{P}(a,b)\,
\log\!\left(
\frac{\mathbb{P}(a,b)}{\mathbb{P}(a)\,\mathbb{P}(b)}
\right).
\end{align}

\subsection{Calculation of the partition function $Z(\b)$}

We compute the thermal partition function
\begin{align}
Z(\beta)=\mathrm{Tr}\!\left(e^{-\beta H}\right).
\end{align}
Since the Fock space decomposes into even/odd parity eigenspaces, we can split $H=H_{\mathcal{P}=+}\oplus H_{\mathcal{P}=-}$, hence
\begin{align}
Z(\beta)=\mathrm{Tr}_{\mathrm{even}} e^{-\beta H_{+}} + \mathrm{Tr}_{\mathrm{odd}} e^{-\beta H_{-}},
\end{align}
where
\begin{align}
\mathrm{Tr}_{\mathrm{even}}(\cdots)=\tfrac12\left(\mathrm{Tr}(\cdots)+\mathrm{Tr}\!\left[(-1)^{N_f}(\cdots)\right]\right),
\qquad
\mathrm{Tr}_{\mathrm{odd}}(\cdots)=\tfrac12\left(\mathrm{Tr}(\cdots)-\mathrm{Tr}\!\left[(-1)^{N_f}(\cdots)\right]\right).
\end{align}
When $\mathcal{P}=+1$, the boundary term has $-J(c_n^\dagger c_1+\mathrm{h.c.})$ giving periodic (P) boundary conditions for the fermions, and when $\mathcal{P}=-1$ we get anti-periodic (AP) boundary conditions.
Thus $H_{+}\equiv H_{P}$ and $H_{-}\equiv H_{AP}$. Generalize the boundary conditions by
\begin{align}
c_{n+1}=e^{i\theta}c_1,
\end{align}
where $\theta=0$ (P) and $\theta=\pi$ (AP).
Define the Fourier transform
\begin{align}
c_j=\frac{1}{\sqrt{n}}\sum_{k\in K_\theta} e^{ikj}\,c_k,
\end{align}
with
\begin{align}
K_\theta=\left\{\frac{2\pi m+\theta}{n}\right\}_{m=0}^{n-1}.
\end{align}
The $k$-space Hamiltonian is
\begin{align}
H_{\theta}
&=\sum_{k\in K_{\theta}}
\left(
2(h-J\cos k)\left(c^\dagger_k c_k-\tfrac12\right)
+J\sin k\,(c_{-k}c_k+c^\dagger_{k}c^\dagger_{-k})
\right).
\end{align}
For $k\neq 0,\pi$ (so that $k\not\equiv -k\ \mathrm{mod}\ 2\pi$), define the Nambu spinor
\begin{align}
\Phi_k=\begin{pmatrix} c_k\\ c_{-k}^\dagger\end{pmatrix},
\end{align}
and set $\xi_k=2(h-J\cos k)$ and $\Delta_k=2J\sin k$.
Then the combined $(k,-k)$ contribution becomes
\begin{align}\label{eq:paired}
H_{k,-k}=\Phi_k^\dagger
\begin{pmatrix}
\xi_k & \Delta_k\\
\Delta_k & -\xi_k
\end{pmatrix}\Phi_k.
\end{align}
The eigenvalues are $\pm \varepsilon(k)$ with
\begin{align}
\varepsilon(k)=\sqrt{\xi_k^2+\Delta_k^2}
=2\sqrt{(h-J\cos k)^2+(J\sin k)^2}.
\end{align}
At $k=0$ or $k=\pi$, $\Delta_k=0$ and the Hamiltonian reduces to
\begin{align}\label{eq:special}
H_k=\xi_k\left(c_k^\dagger c_k-\tfrac12\right).
\end{align}
Introducing Bogoliubov fermions $\gamma_{\pm k}$ for the blocks \cref{eq:paired} gives
\begin{align}
H_{k,-k}=\varepsilon(k)\left(\gamma^\dagger_k\gamma_k+\gamma^\dagger_{-k}\gamma_{-k}-1\right).
\end{align}
For the special modes \cref{eq:special}, take $\gamma_k=c_k$.
Altogether,
\begin{align}
H_{\theta}=\sum_{k\in K_\theta}\varepsilon(k)\left(\gamma_k^\dagger\gamma_k-\tfrac12\right).
\end{align}
Tracing over a single mode yields
\begin{align}
\mathrm{Tr}_{n=0,1}\, e^{-\beta \varepsilon (n-\tfrac12)} = 2\cosh\!\left(\frac{\beta \varepsilon}{2}\right),
\end{align}
so the full trace factorizes:
\begin{align}
\mathrm{Tr}\, e^{-\beta H_{\theta}}
= \prod_{k\in K_{\theta}} 2\cosh\!\left(\frac{\beta \varepsilon(k)}{2}\right)
\equiv C_{\theta}(\beta).
\end{align}
The parity-twisted trace inserts $(-1)^{N_f}$.
Define $N_\gamma=\sum_{k\in K_\theta}\gamma_k^\dagger\gamma_k$.
Then the physical fermion parity differs from the quasiparticle parity by the Bogoliubov-vacuum parity $P_0^\theta=\pm 1$:
\begin{align}
(-1)^{N_f}=P_0^\theta\,(-1)^{N_\gamma}.
\end{align}
For a single mode,
\begin{align}
\mathrm{Tr}_{n=0,1}\left[(-1)^n e^{-\beta \varepsilon (n-\tfrac12)}\right]
=2\sinh\!\left(\frac{\beta \varepsilon}{2}\right),
\end{align}
hence
\begin{align}
\mathrm{Tr}\left[(-1)^{N_f}e^{-\beta H_\theta}\right]
= P_0^\theta \prod_{k\in K_\theta} 2\sinh\!\left(\frac{\beta \varepsilon(k)}{2}\right)
\equiv P_0^\theta\, S_\theta(\beta).
\end{align}
Thus the full partition function is
\begin{align}\label{eq:Z_final}
Z(\beta)
&=
\frac12\Big(C_0(\beta)+\mathrm{sgn}(h-J)\,S_0(\beta)\Big)
+\frac12\Big(C_\pi(\beta)-S_\pi(\beta)\Big),
\end{align}
where
\begin{align}
C_\theta(\beta)&=\prod_{k\in K_{\theta}}2\cosh\!\left(\frac{\beta \varepsilon(k)}{2}\right),
\qquad
S_\theta(\beta)=\prod_{k\in K_{\theta}}2\sinh\!\left(\frac{\beta \varepsilon(k)}{2}\right),\\
\varepsilon(k)&=
2\sqrt{(h-J\cos k)^2+(J\sin k)^2},
\end{align}
and
\begin{align}
K_0=\left\{\frac{2\pi m}{n}\right\}_{m=0}^{n-1},
\qquad
K_\pi=\left\{\frac{2\pi(m+\tfrac12)}{n}\right\}_{m=0}^{n-1}.
\end{align}

\subsection{Final expression for $\frac12$-cut CMI}

Let $a,b\in\{0,1\}^n$ label computational basis states (bitstrings), and let
$I(a)\subset\{1,\dots,n\}$ and $J(b)\subset\{1,\dots,n\}$ be their occupied-site sets.
Define
\begin{align}
S(a,b):=J(b)\ \cup\ \big(n+I(a)\big),
\end{align}
as in \cref{eq:S-set}.
Let $\sigma(a):=(-1)^{|a|}$ be the parity of $a$ (even $\Rightarrow \sigma=+1$, odd $\Rightarrow \sigma=-1$). Define the (unnormalized) joint weight
\begin{align}
w(a,b):=
\begin{cases}
\Big|\det\big(T_{22}^{(+)}(\beta/2)\big)\Big|
\;\Big|\mathrm{pf}\!\big(\mathcal{A}_{+}^{(\beta/2)}[S(a,b),S(a,b)]\big)\Big|^2,
& |a|,|b|\ \text{even},\\[4pt]
\Big|\det\big(T_{22}^{(-)}(\beta/2)\big)\Big|
\;\Big|\mathrm{pf}\!\big(\mathcal{A}_{-}^{(\beta/2)}[S(a,b),S(a,b)]\big)\Big|^2,
& |a|,|b|\ \text{odd},\\[4pt]
0,& |a|\not\equiv |b|\ (\mathrm{mod}\ 2).
\end{cases}
\label{eq:w_joint}
\end{align}
Then the dephased joint probability is
\begin{align}
\mathbb{P}(a,b)=\frac{w(a,b)}{Z(\beta)},
\end{align}
with $Z(\beta)$ given in \cref{eq:Z_final}. The marginals can be written without an explicit sum using completeness:
\begin{align}
\mathbb{P}_A(a)
&=\sum_b \mathbb{P}(a,b)
=\frac{1}{Z(\beta)}\sum_b \big|\langle b|e^{-\frac{\beta}{2}H_{\sigma(a)}}|a\rangle\big|^2
=\frac{1}{Z(\beta)}\langle a|e^{-\beta H_{\sigma(a)}}|a\rangle,
\label{eq:marginal_A}
\\
\mathbb{P}_B(b)
&=\sum_a \mathbb{P}(a,b)
=\frac{1}{Z(\beta)}\langle b|e^{-\beta H_{\sigma(b)}}|b\rangle.
\label{eq:marginal_B}
\end{align}
Each diagonal matrix element $\langle a|e^{-\beta H_{\pm}}|a\rangle$ is given by the same Pfaffian theorem
\cref{eq:kernel-pfaffian} with $\beta/2\to \beta$ and $I=J=I(a)$. Finally, substituting $\mathbb{P}(a,b)=w(a,b)/Z(\beta)$ into \cref{eq:CMI-alt} gives
\begin{align}
I_{\mathrm{cl}}(A{:}B)
&=
\log Z(\beta)
+\sum_{a,b}\frac{w(a,b)}{Z(\beta)}
\Big(
\log w(a,b)-\log w_A(a)-\log w_B(b)
\Big),
\label{eq:Icl_final_weights}
\end{align}
where
\begin{align}
w_A(a):=\langle a|e^{-\beta H_{\sigma(a)}}|a\rangle,
\qquad
w_B(b):=\langle b|e^{-\beta H_{\sigma(b)}}|b\rangle.
\end{align}

\begin{figure}[h]
    \centering
    \includegraphics[width=0.5\linewidth]{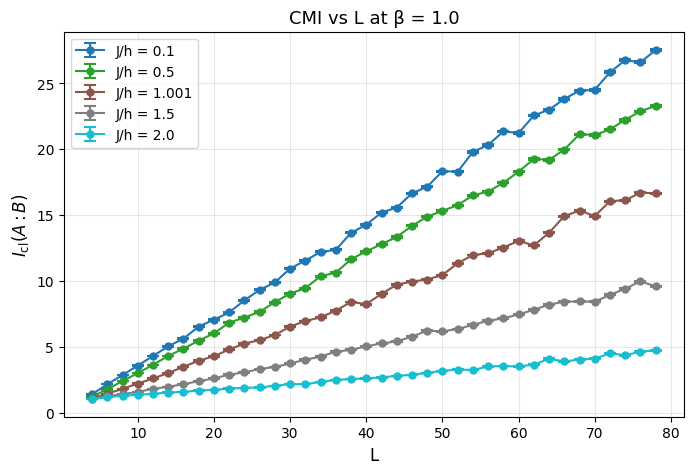}
    \label{fig:placeholder}
    \caption{Classical mutual information of the $Z$-dephased TFD state for the transverse-field Ising model at inverse temperature $\beta =1$ obtained using MCMC importance sampling. All curves exhibit volume-law scaling $I_{\mathrm{cl}}(A{:}B) \sim c(\beta,J/h) \cdot n$, for different values of the coupling ratios $J/h.$}
\end{figure}

\subsection{The limiting behavior of the $\frac12$-cut dephased mutual information}
\subsubsection{$\beta\to 0$ and small $\beta$ expansion}

We have
\begin{align}
\braket{a,b}{\mathrm{TFD}_\beta}
= \frac{1}{\sqrt{Z(\beta)}}\bra{a}e^{-\beta H/2}\ket{b}
= \frac{1}{\sqrt{Z(\beta)}}K(a,b).
\end{align}
In the limit $\beta\to 0$, the kernel reduces to the overlap
\begin{align}
K_0(a,b)=\bra{a}\mathbb{I}\ket{b}=\delta_{a,b}.
\end{align}
Furthermore, $Z(0)=\mathrm{Tr}(\mathbb{I})=2^n$. Thus
\begin{align}
\mathbb{P}_0(a,b)=\frac{\delta_{a,b}}{2^n}.
\end{align}
The marginals are uniform,
\begin{align}
\mathbb{P}_0(a)=\mathbb{P}_0(b)=\frac{1}{2^n},
\end{align}
hence (with $\log$ base $2$)
\begin{align}
I(A{:}B)=H(A)+H(B)-H(A,B)=n+n-n=n.
\end{align}
\subsubsection{Small-$\beta$ expansion}
Fix a parity sector $\sigma \in\{+,-\}$ and define
\begin{align}
    K_{\sigma}(\beta) = e^{-\frac{\beta}{2}H_{\sigma}}\,, \qquad K_{\sigma}(b,a;\beta)\;:=\;\langle b|K_\sigma(\beta)|a\rangle.
\end{align}
We can expand the matrix-exponential in powers of $\beta$ in the limit that $\beta\rightarrow0$ as follows
\begin{align}
    K_\sigma(\beta)
=
\mathbb I - \frac{\beta}{2}H_\sigma + \frac{\beta^2}{8}H_\sigma^2 +\mathcal{O}(\beta^3),
\end{align}
the corresponding matrix element for any basis states $a,b$ in the same parity sector $\sigma$ becomes
\begin{align}
K_{\sigma}(b,a;\beta)
=
\delta_{ab}
-\frac{\beta}{2}(H_\sigma)_{ba}
+\frac{\beta^2}{8}(H_\sigma^2)_{ba}
+ \mathcal{O}(\beta^3).
\label{eq:Kexp}
\end{align}
When $a \neq b,$ we have
\begin{align}
K_{\sigma}(b,a;\beta) &= -\frac{\beta}{2}(H_\sigma)_{ba} + \mathcal O(\beta^2), \nonumber \\\Rightarrow
w(a,b)&=\frac{\beta^2}{4}\,(H_\sigma)_{ba}^2 + \mathcal O(\beta^3),
\label{eq:w_offdiag}
\end{align}
and for the diagonal case $a=b,$ setting $x=-(\beta/2)(H_\sigma)_{aa}$ and $y=(\beta^2/8)(H_\sigma^2)_{aa}$,
\begin{align}
K_{\sigma}(a,a;\beta)&=1+x+y+\mathcal O(\beta^3)
\Longrightarrow
w(a,a)=1-\beta(H_\sigma)_{aa}
+\frac{\beta^2}{4}\Big((H_\sigma)_{aa}^2+(H_\sigma^2)_{aa}\Big)
+\mathcal O(\beta^3).
\label{eq:w_diag}
\end{align}
Similarly,
\begin{align}
w_A(a)=\langle a|e^{-\beta H_{\sigma(a)}}|a\rangle
=
1-\beta(H_{\sigma(a)})_{aa}
+\frac{\beta^2}{2}(H_{\sigma(a)}^2)_{aa}
+\mathcal O(\beta^3),
\label{eq:wAexp}
\end{align}
and the same formula holds for $w_B(b)$ with $a\mapsto b$. 
We can get the $\beta$ expansion for the partition function by making use of the completeness relation
\begin{align}
Z(\beta)=\sum_{a,b}w(a,b)=\mathrm{Tr}\big(e^{-\beta H}\big),
\end{align}
so the traced expansion gives
\begin{align}
Z(\beta)=2^n-\beta\,\mathrm{Tr}(H)+\frac{\beta^2}{2}\mathrm{Tr}(H^2)+\mathcal O(\beta^3).
\label{eq:Zexp}
\end{align}
Since $\mathrm{Tr}(H)=0$, the first correction is at $\mathcal O(\beta^2).$ For $a\neq b$ in the same parity sector $\sigma$,
combining \cref{eq:w_offdiag} with $Z(\beta)=2^n+\mathcal O(\beta)$ gives
\begin{align}
\mathbb P(a,b)=\frac{w(a,b)}{Z(\beta)}
=
\frac{\beta^2}{4\,2^n}\,(H_{\sigma(b)})_{ba}^2
+\mathcal O(\beta^3).
\label{eq:Poffdiag}
\end{align}
The marginal probability distribution over $A$ (and $B$) is given by
\begin{align}
    \mathbb{P}_{A}(a) = \frac{1}{2^n} \left( 1 - \beta H_{aa} +\frac{\beta^2}{2} (H^2)_{aa} -\frac{\beta^2}{2 \,2^n} \Tr(H^2) \right) + \mathcal{O}(\beta^3)\,.
\end{align}
Define $\mathbb P_{A}(a) = 2^{-n}(1+\delta_a)$ where $\delta_a$ includes the $\mathcal{O} (\beta)$ and $\mathcal{O} (\beta^2)$ terms above and notice that $\sum_a \delta_a = 0.$ Using $\log_2 \mathbb P_A(a) = -n + \frac{1}{\ln 2} \ln (1+\delta_a),$ we get
\begin{align}
    H(A) &= -\sum_a 2^{-n} (1+\delta_a) \left( -n + \frac{1}{\ln 2} \ln (1+\delta_a) \right) + \mathcal{O}(\beta^3) \nonumber \\
    &=n - \frac{1}{2^n\ln 2}\sum_a (1+\delta_a) \ln (1+\delta_a)+ \mathcal{O}(\beta^3) \nonumber \\
    &= n - \frac{1}{2^n\ln 2}\sum_a \left( \delta_a + \frac{\delta_a^2}{2} \right)+ \mathcal{O}(\beta^3) \nonumber \\
    &= n - \frac{1}{2^n\ln 2}\sum_a  \frac{\delta_a^2}{2} + \mathcal{O}(\beta^3) \,,
\end{align}
where we have used the expansion
\begin{align}
    (1+\delta_a) \ln (1+\delta_a) = \delta_a + \frac{\delta_a^2}{2} + \mathcal{O}(\delta_a^3)\,,
\end{align}
since $\delta$ is small when $\beta$ is small. Keeping only the $\mathcal{O}(\beta^2)$ term in $\delta^2_a,$ we get
\begin{align}\label{eq:HA}
    H(A) = n - \frac{\beta^2}{2 \,2^n\ln 2}\sum_a (H_{\sigma(a)})_{aa}^2   + \mathcal{O}(\beta^3)\,. 
\end{align}
Note that the off-diagonal probabilities in \cref{eq:Poffdiag} are order $\beta^2$, which is the key source of the leading correction to the entropy.
At $\beta=0$,
\begin{align}
K_\sigma(0)=\mathbb I\quad\Rightarrow\quad
\mathbb P_0(a,b)=2^{-n}\delta_{ab}.
\end{align}
Hence the marginals are uniform, $\mathbb P_{A,0}(a)=\mathbb P_{B,0}(b)=2^{-n}$, while the joint is supported only on $a=b$.
For $\beta>0$ small, \cref{eq:Poffdiag} implies that probability leaks from the diagonal to off-diagonal entries only at order $\beta^2$. Fix $b$ and define the conditional distribution $\mathbb P(\cdot|b)=\mathbb P(\cdot,b)/\mathbb P_B(b)$.
At $\beta=0$, $\mathbb P_0(a|b)=\delta_{ab}$, so $\mathbb P(\cdot|b)$ is a delta distribution.
For small $\beta$, define the leak
\begin{align}
\varepsilon_b
:=
\sum_{\substack{a\neq b\\ \sigma(a)=\sigma(b)}}\mathbb P(a|b).
\end{align}
Then $\mathbb P(b|b)=1-\varepsilon_b$ and $\varepsilon_b=\mathcal O(\beta^2)$.
If we list the nonzero values of $\mathbb P(\cdot|b)$ as a vector, it has the form
\begin{align}
p_0=1-\varepsilon_b,\qquad p_i=q_i\ (i=1,\dots,m),\qquad \sum_{i=1}^m q_i=\varepsilon_b,
\label{eq:near_det_vec2}
\end{align}
where the $q_i$'s are the off-diagonal conditional probabilities. For any such vector, the Shannon entropy (natural logs) satisfies
\begin{align}
H(p)
=-(1-\varepsilon)\ln(1-\varepsilon)-\sum_{i=1}^m q_i\ln q_i,
\end{align}
and using $-\ln(1-\varepsilon)=\varepsilon+\mathcal O(\varepsilon^2)$ gives the expansion
\begin{align}
H(p)
=
\sum_{i=1}^m q_i\big(1-\ln q_i\big)+\mathcal O(\varepsilon^2),
\label{eq:entropy_near_det2}
\end{align}
valid as $\varepsilon\to 0$. 
Using $\mathbb P_B(b)=2^{-n}+\mathcal O(\beta)$ and \cref{eq:Poffdiag},
for $a\neq b$ with $\sigma(a)=\sigma(b)$,
\begin{align}
\mathbb P(a|b)
=
\frac{\mathbb P(a,b)}{\mathbb P_B(b)}
=
\frac{\beta^2}{4}\,|(H_{\sigma(b)})_{ab}|^2+\mathcal O(\beta^3),
\label{eq:cond_off2}
\end{align}
so
\begin{align}
\varepsilon_b
=
\frac{\beta^2}{4}\sum_{\substack{a\neq b\\ \sigma(a)=\sigma(b)}}|(H_{\sigma(b)})_{ab}|^2+\mathcal O(\beta^3).
\label{eq:epsb2}
\end{align}
Applying \cref{eq:entropy_near_det2} to $\mathbb P(\cdot|b)$ and using \cref{eq:cond_off2} yields
\begin{align}
H(A|B=b)
&=
\sum_{\substack{a\neq b\\ \sigma(a)=\sigma(b)}}
\mathbb P(a|b)\Big(1-\ln \mathbb P(a|b)\Big)
+\mathcal O(\varepsilon_b^2)
\nonumber\\
&=
\frac{\beta^2}{4}
\sum_{\substack{a\neq b\\ \sigma(a)=\sigma(b)}}
|(H_{\sigma(b)})_{ab}|^2
\Bigg(
1-\ln\Big(\frac{\beta^2}{4}|(H_{\sigma(b)})_{ab}|^2\Big)
\Bigg)
+\mathcal O(\beta^3|\ln\beta|)+\mathcal O(\beta^4).
\label{eq:Hab_given_b2}
\end{align}
Averaging over $b$ gives (in base-$2$)
\begin{align}
H(A|B)
&=\sum_b \mathbb P_B(b)\,H(A|B=b)
\nonumber\\
&=
\frac{\beta^2}{4\,2^n}
\sum_b\sum_{\substack{a\neq b\\ \sigma(a)=\sigma(b)}}
|(H_{\sigma(b)})_{ab}|^2
\Bigg(
\frac{1}{\ln 2}-\log_2\Big(\frac{\beta^2}{4}(H_{\sigma(b)})_{ab}^2\Big)
\Bigg)
+\mathcal O(\beta^3|\ln\beta|)+\mathcal O(\beta^3).
\label{eq:Hab2}
\end{align}
The classical mutual information is given by
\begin{align}
I_{\rm cl}(A{:}B)=H(A)+H(B)-H(A,B)=H(A)-H(A|B).
\end{align}
Combining \cref{eq:HA} with \cref{eq:Hab2},
\begin{align}
I_{\rm cl}(A{:}B)
&=
n
-\frac{\beta^2}{2\ln 2}\frac{1}{2^n}\sum_a H_{aa}^2
-\frac{\beta^2}{4\,2^n}\sum_b\sum_{a\neq b}H_{ab}^2
\left[\frac{1}{\ln 2}-\log_2\!\left(\frac{\beta^2}{4}H_{ab}^2\right)\right]
+\mathcal O(\beta^3|\log\beta|)\,.
\end{align}
For $H=-J\sum_i X_iX_{i+1}-h\sum_i Z_i$, one has
$H_{aa}=-h\sum_i z_i(a)$ with $z_i(a)=\pm1$, hence
\begin{align}
\frac{1}{N}\sum_a H_{aa}^2 = h^2 n.
\end{align}
Also the only off-diagonal elements come from $-JX_iX_{i+1}$, giving
$H_{ab}^2=J^2$ for exactly $n$ neighbors $a\neq b$ of each $b$, hence
\begin{align}
\sum_b\sum_{a\neq b}H_{ab}^2 = 2^n n J^2,
\qquad
\log_2\!\left(\frac{\beta^2}{4}H_{ab}^2\right)=\log_2\!\left(\frac{\beta^2J^2}{4}\right).
\end{align}
Therefore,
\begin{align}
I_{\rm cl}(A{:}B)
&=
n
-\frac{\beta^2}{2\ln 2}\,h^2 n
-\frac{\beta^2}{4}\,J^2 n
\left[\frac{1}{\ln 2}-\log_2\!\left(\frac{\beta^2J^2}{4}\right)\right]
+\mathcal O(\beta^3|\log\beta|)\,.
\end{align}

\end{document}